\documentclass[11pt]{article}
\usepackage{amsmath, amsthm, amssymb, bm}
\usepackage[a4paper, total={6in, 10in}]{geometry}
\usepackage{graphicx, psfrag, multicol}
\usepackage{setspace}
\usepackage{natbib}
\usepackage{dsfont}
\usepackage[nottoc]{tocbibind}
\newcommand{\indep}{\perp \!\!\! \perp}
\usepackage{multirow}
\usepackage{multicol}
\usepackage{graphicx}
\usepackage{caption}
\usepackage{subcaption}
\usepackage{lineno}
\usepackage{float}
\usepackage{hyperref}

\newtheorem{proposition}{Proposition}

\title{Disentangling 
spatial interference and spatial confounding biases in causal inference.} 

\author{Isqeel Ogunsola$^{1,2}$ and Olatunji Johnson$^1$}

\date{$^1$Department of Statistics, The University of Manchester, United Kingdom. \\
$^2$Department of Statistics, Federal University of Agriculture, Abeokuta, Nigeria} 

\begin{document}
\maketitle
\onehalfspacing
\begin{center}
\textbf{Abstract}
\end{center}
Spatial interference and spatial confounding are two major issues inhibiting precise causal estimates when dealing with observational spatial data. Moreover, the definition and interpretation of spatial confounding remain arguable in the literature. In this paper, our goal is to provide clarity in a novel way on misconception and issues around spatial confounding from Directed Acyclic Graph (DAG) perspective and to disentangle both direct, indirect spatial confounding and spatial interference based on bias induced on causal estimates. Also, existing analyses of spatial confounding bias typically rely on Normality assumptions for treatments and confounders, assumptions that are often violated in practice. Relaxing these assumptions, we derive analytical expressions for spatial confounding bias under more general distributional settings using Poisson as example . We showed that the choice of spatial weights, the distribution of the treatment, and the magnitude of interference critically determine the extent of bias due to spatial interference. We further demonstrate that direct and indirect spatial confounding can be disentangled, with both the weight matrix and the nature of exposure playing central roles in determining the magnitude of indirect bias. Theoretical results are supported by simulation studies and an application to real-world spatial data. In future, parametric frameworks for concomitantly adjusting for spatial interference, direct and indirect spatial confounding for both direct and mediated effects estimation will be developed.

\vspace{2mm}
\noindent Keywords: Observational spatial data, Bias, Directed Acyclic Graph indirect spatial confounding, spatial weights

\section{Introduction}

\subsection{Spatial interference}

Spatial interference is a phenomenon where treatment or exposure in one location affects the response in another location. This phenomenon is often neglected by researchers, therefore assuming no spatial interference (for example, \citet{papadogeorgou2019}; \citet{hernan2010causal}) for ease of identification and treatment effect estimation \citep{imbens2020}. This assumption, together with the no multiple treatment assumption, is what Rubin called the Stable Unit Treatment Value Assumption \citep{rubin1980}. However, in a spatial setting, no interference assumption may be violated as exposure in one location could affect the outcome in nearby locations due to spatial dependency. Spatial interference can take different forms, most commonly clustered (or partial) and general interference. In clustered interference \citep{sobel2006}, units are partitioned such that only units in the same cluster are assumed to interact or interfere with one another. In contrast, general interference \citep{giffin2020b} places no structural constraints, allowing units to interfere freely. While such interference effects can be identified under certain conditions \citep{tchetgen2012}, the required data are rarely available in real-world applications \citep{giffin2020b}.

Authors have proposed several approaches for estimating spatial interference effects for observational data. An approach based on generalising propensity score to include interference effects as part of the covariates in the propensity score model was proposed by \citet{giffin2020} for estimating direct and spillover effects, but they assumed no confounding effects. They adopted Bayesian spline regression to avoid model dimension issues. The use of spatial instrumental variables was also employed by \citet{giffin2021} in estimating interference effects. Though, they took confounding effect into consideration but finding a suitable spatial instrumental variable and testing/validating some of the assumptions can be a herculean task if not impossible. In the presence of interference, \citet{zigler2020} formulated different possible causal estimands for two different cases of interference. The first case is when the treatment at one unit is different from that at the units where the outcomes are measured. The second case is when there is interaction between the units, i.e units at one location affect the outcome at another location. They utilised the inverse probability weighting methods to estimate a simplified subset of the estimands.  Novel inverse probability weighting estimators were derived by \citet{tchetgen2012} when the spatial interference effect is present. Regression discontinuity was also employed in causal effects estimation in the presence of interference by (\citet{torrione2024}). \citet{torrione2024} proposed a nonparametric estimator for the causal estimands based on distance and obtained the properties. The model assumed conditional unconfoundedness. Estimating treatment effect in the presence of interference is not limited only to observational data. It has also been addressed in randomised trials (e.g. \citet{basse2019};  \citet{athey2018}; \citet{sobel2006}). 

These researchers make attempts to estimate interference effects but failed to quantify and evaluate the impacts of the bias induced on the causal estimates when spatial interference is present but not accounted for and this is one contribution of this study.

\subsection{Spatial confounding}

On the other hand, the definition of spatial confounding is debatable in the literature. It had been meant to be mainly two different things (\citet{papadogeorgou2023}, \citet{urdangarin2023}). Firstly, it was assumed to be a situation where the covariate is correlated with the spatial random effect in a regression model (\citet{dupont2022}, \citet{reich2006}, \citet{khan2022}). Secondly, it had been referred to as a situation in causal inference where the measured covariates do not account for the missing confounder, and the unmeasured confounder has a spatial structure (\citet{gilbert2021}, \citet{papadogeorgou2023}). Meanwhile \citet{gilbert2024} stated that omitted confounder bias, random effect collinearity, regularisation bias and concurvity are all sometimes referred to as spatial confounding in literature. Recently, \citet{khan2023} tried to distinguish all the terms referred to as spatial confounding by classifying the phenomenon into two: referring to the former as analysis model spatial confounding and the latter as data generation spatial confounding. As contribution and further clarity, we set a clear difference between spatial confounding in causal inference and other situations from Directed Acylic Ggraph (DAG) perspective in later section. The two cases stated by  \citet{khan2023} had received numerous attention over years. For the former, among the prominent approaches are: spatial+ method (\citet{dupont2022}), restricted spatial regression method (\citet{reich2006}, \citet{khan2022}), and Transformed Gaussian Markov Random Field (TGMRF) approach (\citet{banerjee2014} and \cite{prates2015}) to mention a few. These methods are based on replacement of the spatial random effects with another forms of random effects with the aim of removing collinearity between the covariates and the variance-covariance matrix of the new random effects. Authors have also compared the performances of some of these methods and determined which produces the best fixed effect estimates (under certain conditions) in the presence of spatial confounders. For further details about these comparisons, see \citet{urdangarin2023}. However, the main focus here is the latter, i.e. spatial confounding as it relates to causal inference. The assumption of no confounder in causal inference does not hold for observational data and becomes complicated in a spatial setting. 

Several researchers have also proposed methods to mitigate this challenge in causal inference. Recently, \citet{gilbert2021} proposed an approach for causal effect estimation independent of distribution assumption. They use spatial coordinates to adjust for unmeasured confounders with confounder measurability as space function and lack of spatial components in the exposures as two additional conditions. They adopted double machine learning procedure in estimating the treatment effect of interest. An exposure-penalised spline was also proposed by \citet{bobb2022} in reducing spatial confounding effects. They explain the spatial variability of exposure through the degree of spatial smoothing. In some causal inference literature, a spatial random effect is used to represent spatial confounder (e.g. \citet{reich2006}). However, authors have argued that using spatial random effects as a substitute for spatial confounder does not mitigate the problem of unobserved spatial confounding (\citet{schnell2020}; \citet{paciorek2010}. \citet{schnell2020} then developed an affine estimator for addressing spatial confounding based on bias term of commonly used linear least square estimators. The method was implemented using a Bayesian approach. The geoadditive structural equation model was also proposed by (\citet{thaden2018}) in mitigating spatial confounding. Specifically, they quantified spatial effects on both the response and covariates at the same time and also incorporated discrete spatial information in their method development. 
\citet{papadogeorgou2019} employed adjusted distance as a form of including information on spatial proximity of units in propensity scores computation in ameliorating the spatial confounding problem. The study showed that their adjusted approach was better than other methods of spatial information inclusion in propensity score adjustment.  

Due to several misconception and lack of consensus on solution to spatial confounding, authors have also tried to quantify and evaluate the bias so as to suggest possible solutions. \citet{dupont2023} quantified the bias introduced by spatial confounder and stated that this is not negligible when the frequency is high (i.e. weights close to 1), which means strong spatial correlation. \citet{narcisi2024} examined and evaluated spatial confounding bias using the theory of quadratic forms.  Their study revealed that both the marginal covariate variability and the marginal covariance between the confounder and the covariate are key determinants of the magnitude of bias due to confounding. \citet{keller2020} developed an approach based on splines for confounding adjustment using Fourier and wavelet filtering, circumventing the difficulty in spatial scale interpretation posed by the spline degree of freedom. \citet{paciorek2010} also showed that confounding bias depends on the spatial scales of the residuals and covariates. He also deduced that, when the spatial scale of the covariate is smaller than the scale of the unmeasured confounding, fitting a random effects model can help reduce the bias. Meanwhile, none of the existing studies quantifies the bias due to indirect spatial confounding and assumed both to be the same which is not true. Also, quantifying bias due spatial confounding when the treatment and confounder distribution is not normal has also been left handicapped.


\subsection{Clarity on spatial confounding from Directed Acyclic Graph (DAG) perspective.}
The Directed Acylic Graph mostly shortened as DAG is a diagram used in representing and visualizing causal relationships among variables. In Figure \ref{fig DAG diff}, $A, Y $ and $U$ represents the treatment, response and confounder respectively. The nodes represent the variables and it is directed due to the fact that an arrow from one variable, for example  $A$  to $Y$  shows that $A$ has a causal effect on $Y$ and not vice versa. The graph being "acyclic" implies that there are no cycles in the DAG. This means that $A$, for instance cannot affect itself in anyway. Apart from being used in modelling causal relationships, DAGs also make it easier to identify which variable to condition on (i.e. adjusting for variables that block non-causal paths), helps in clarifying causal assumptions and figuring out possible confounders for the exposure and response. Note that the box on $U$ indicate that $U$ is missing and has a spatial structure.  

In this section, we provided further clarity to issues and misconception about spatial confounding from DAG perspective.

\begin{figure}[H]
\begin{subfigure}{0.6\textwidth}
\includegraphics[width=3.6cm]{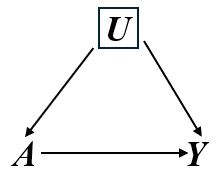}
 \caption{Spatial confounding in causal inference}
 \label{fig DAG causal}
\end{subfigure}
\begin{subfigure}{0.4\textwidth}
\includegraphics[width=3.6cm]{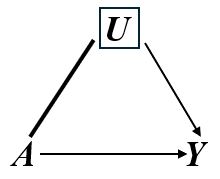}
 \caption{Other spatial confounding}
  \label{fig DAG non causal}
\end{subfigure}
\caption{Spatial confounding}
\label{fig DAG diff}
\end{figure}

Figures \ref{fig DAG causal} and \ref{fig DAG non causal} show the DAGs for spatial confounding in causal inference and non causal inference. The arrows from $U$ to $Y$ and from $A$ to $Y$ in the two figures show that both $U$ and $A$ has causal effects on $Y$. However, in non causal inference there exist relationship between the spatial missing variable $U$ and treatment $A$ but these relationship is not causal indicated by a non-head arrow between $U$ and $A$ as shown in Figure \ref{fig DAG non causal}. On the other hand, the relationship between $U$ and $A$ in Figure \ref{fig DAG causal} is causal, represented with an arrow from $U$ to $A$. Being causal here means that it is not just an association but the conditions/assumptions stated in section 2 are satisfied for $U$ on $A$. For example, in Figure \ref{fig DAG causal}, $U$ must occur before $A$ but this is not necessary in \ref{fig DAG non causal}. In short, the clear difference between spatial confounding in causal inference and non causal inference is that the relationship between the missing spatial variable and $Y$ and $A$ has causal interpretation but this is not so with respect to spatial confounding in non causal inference. 

Meanwhile, spatial confounding can arise in two scenarios: (i) direct spatial confounding, where an unmeasured spatial variable influences both treatment and outcome at the same location, and (ii) indirect spatial confounding, where an unmeasured spatial variable at one location affects outcomes and/or exposure at other locations, often due to spatial interference. For example, in studying the effects of climate change on health, unmeasured pollution that varies spatially can affect both climate change and health at the same location (direct spatial confounding). On the other hand, the pollution in one location could cause climate change and also affect health in nearby locations (indirect spatial confounding). While much attention had been given to quantifying and evaluating bias from direct spatial confounding, the impact of indirect spatial confounding and spatial interference remains underexplored. The subtle difference between direct and indirect spatial confounding is illustrated in Figures \ref{figaa} and \ref{figbb} where $U, A,Y$ are spatial confounder, exposure and responses respectively and the subscripts 1 and 2 indicate locations 1 and 2.

\begin{figure}[H]
 \centering
   \includegraphics[width=8cm]{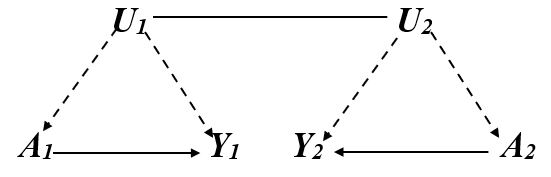}
  \caption{Direct Spatial Confounding}
  \label{figaa}
  
\end{figure}

\begin{figure}[H]
 \centering
   \includegraphics[width=9cm]{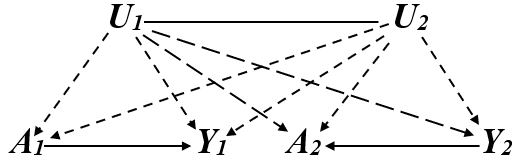}
  \caption{Indirect Spatial Confounding}
    \label{figbb}
\end{figure}

In Figure \ref{figaa} the spatial confounder $U_1$ only affects both the response and exposure at location 1 and $U_2$ affects the responses and exposure at location 2 only. The thick line between $U_1$ and $U_2$ indicates spatial dependency between the confounders. In contrast, in Figure \ref{figbb}, the confounders $U_1$ affects exposures and responses at both locations and similarly for $U_2$. The broken line with arrow from  $U_1$ and  $U_2$ to the responses and exposures indicated that the confounder is unmeasured. 

To this end, our objectives are to: (i) quantify the magnitude of the bias induced by spatial interference and indirect spatial confounding analytical (ii) evaluate the bias in different scenario using simulation study (iii) derive the bias for different sampling distribution of the treatment and confounder 
(iv) propose a simple method for bias correction due to spatial interference and (v) clarify issue on spatial confounding from DAG perspective.
Next, we extended our approach to case
when both spatial interference and spatial confounding are concomitantly present. 

Succinctly put, the contributions of this paper are as follows:
\begin{itemize}
    \item [(1)] This research clarifies issues on spatial confounding from Direct Acyclic Graph (DAG) perspective.
 \item [(2)] This is first paper to the best our knowledge,  to quantify and evaluate bias due to spatial interference in both spatial and non spatial settings in causal inference.
 \item [(3)] It disentangled direct and indirect spatial confounding biases against assumption made about both to be one in existing causal literature and also derived the bias in the presence of indirect spatial confounding, a result not in previously in existence in causal inference literature.
 \item [(4)] This work also derived the bias when the distribution of the treatment and confounder are not normally distributed. Specifically, when both are Poisson distributed - a situation not previously looked into. 
 \item [(5)] Finally, this study showed that when any of the three phenomena (spatial interference, direct and indirect spatial confounding) are neglected, treatment effect estimate are biased and decisions made based on the estimates will not be reliable.   
\end{itemize}

Following the introduction and contributions of this research presented in section 1, the assumptions made for our estimate to be valid causal estimate and analytical framework are given in sections 2 and 3 respectively. We derived the bias under different scenarios in section 4. In section 5, we evaluated the bias using simulation, presented the results and gave real data application in section 6. The discussion, simple suggestion on bias correction due to spatial interference only, limitations of this study and future work are given in section 7. In the last section, we gave conclusion about the findings of the study.

\section{Causal effect identification assumptions}
To ensure that the direct estimate is causal and identifiable, the following assumptions must hold. 

Suppose that $Y, A $, $X$ and are the outcome, exposure and covariates, respectively, and  $ "a" $ represents the level of the treatment when it is given.

\noindent \textbf{Assumption 1 (Positivity) :} $Pr[A = a | X = x] > 0 $, where $f_X (x) > 0$. This means that each of the units in the study has a non-zero probability of being assigned to the exposure.

\noindent \textbf{Assumption 2 (Consistency) :} If $A = a,$ then $Y^a = Y$. Consistency here means that if a unit/subject receives a particular level of exposure/treatment,  the observed response must be equal to the response under that exposure. This implies that treatment should be well defined and uniform across all units/subjects.


\noindent \textbf{Assumption 3 (Exchangeability) :} $Y^a \indep  A | X $. This assumptions means that there should be no significant difference in the baseline characteristics of the groups compared (e.g., for binary case, i.e., the treated and the untreated), so that the difference in the average potential outcomes between the treated group and the untreated groups could be attributed to the exposure effect. In short, this simply means that given observed covariates, treatment should be independent on the potential outcome. This is also referred to as unconfoundeness assumption.  Exchangeabily assumption must hold for observational data to mimic randomisation techniques used in experimental data for valid causal inference.

\noindent \textbf{Assumption 4 (Temporal order):} This assumption gives one clear difference between causation and association. For validity of our causal effetcs estimates, we assumed that the exposure here precedes the outcome and not vice-versa.

Note that, in this section we assumed that no confounding assumption holds but the no interference assumption is violated. In the later section, the unconfoundeness assumption will also be relaxed so that both are violated. 

\section{Analytical framework}
First, we started with a data generating model with non-spatial structure but with spatial treatment given as: 

\begin{equation}
 \textbf{M1} \quad \quad  \quad \quad \quad \quad \quad \quad Y\bf{(s_i)} = \beta_0 + \beta_a A\bf{(s_i)} + \epsilon_i 
\end{equation}

\noindent where $Y\bf{(s_i)}$ are the responses at locations $\bf{s_i} = \bf{s_1}, ..., \bf{s_n} \in \mathbb{R}^2$, $A\bf{(s_i)}$ are the treatment varying spatially at locations $\bf{s_i}$, $\epsilon_i$ are non spatial vector of errors assumed to be normally distributed with mean 0 and variance $\sigma^2$ 
i.e. $\epsilon_i \sim N(0, \sigma^2 ) $ and $\beta_0$ and $\beta_a$ are the intercept and treatment effect of interest respectively. We have dropped the covariate and considered the treatment of interest only just for simplicity.  

Due to the spatial dependency exhibited by the treatments, nearby treatments are likely to interfere with one another and the independence assumption is violated introducing interference effect. We then consider a model with interference but assumed no confounder effects given as:

\begin{equation}
 \textbf{M2} \quad \quad  \quad \quad  \quad \quad  Y\bf{(s_i)} = \beta_0 + \beta_a A\bf{(s_i)} + \beta_{\tilde{a}} \tilde{
 A}\bf{(s_i)}+  \epsilon_i, 
    \label{equation 2}
\end{equation}

where $ \beta_{\tilde{a}}$ is the interference effect, $\tilde{A} = \Psi A = \sum_{i \neq j} \psi_{ij} A_j $ is a function of some weights $\Psi \in \mathbb{R}^{n\times n }$ and treatment $A\bf{(s_i)}$. The weights $\Psi$ is a matrix which represents the level of dependency among the treatment but with $0 $ elements in its diagonal so that no treatment will interact with itself. All other notations take their usual meanings as defined in the previous section.

We further considered the case of spatial response but with interference effects only in our simulation that is when the error is $N(0, \Sigma) $ where $ \Sigma = \sigma^2_{\epsilon} \Omega_{\epsilon}(\omega) +\sigma^2_{\epsilon} I$ and  $\Omega_{\epsilon} (\omega) $ is the spatial correlation function with spatial range parameter $\omega$ Meanwhile, spatial correlation in the error could be due to some unmeasured confounders. As an extension, we considered the model with both spatial interference and spatial confounder. The model is given as: 

\begin{equation}
 \textbf{M3} \quad \quad  \quad \quad    Y\bf{(s_i)} = \beta_0 + \beta_a A\bf{(s_i)} + \beta_{\tilde{a}} \tilde{A}\bf{(s_i)}+ \beta_u U\bf{(s_i)} + \epsilon_i 
    \label{equation 3}
\end{equation}

\noindent where $U\bf{(s_i)}$ is an unknown spatial confounder assumed to follow a normal distribution with mean $0$ and variance-covariance $\sigma^2_{u} \Omega_u$ (violation of this assumption is explored in the later section) and $\beta_u$ is the direct confounding effect. The subscripts $u$ and $\epsilon$ indicate the parameters with respect to the confounder and error, respectively, and we assumed that $\Omega$ is modelled using a Mat\'ern correlation function given as: 

\begin{equation*}
    \Omega (d; \omega, v) = \frac{ 2^{1-v}}{\Gamma (v) } \left(\frac{2 \sqrt{v} d }{\omega}\right)^v K_v \left(\frac{2 \sqrt{v} d }{\omega}\right)
\end{equation*}

\noindent where $\omega$ is the spatial range parameter, $d = |s_i - s_j|$, $v$ is the smoothness parameter and $K_v$ is a Bessel function. Explicitly, we have considered a case where $v = 0.5$ which results into exponential correlation function given as $ \exp{\left(\frac{-d}{\omega}\right)} $.

Lastly and to the best of our knowledge, in causal inference literature on bias quantification and evaluation, researchers have only considered a case where the confounder at one location affects the response at the same location, i.e. the model of the form:

\begin{equation}
 \textbf{M4} \quad \quad  \quad \quad   \quad \quad Y\bf{(s_i)} = \beta_0 + \beta_a A\bf{(s_i)} + \beta_u U\bf{(s_i)} + \epsilon_i 
    \label{equation 4}
\end{equation}

We argued that spatial confounding effects could also be indirect where spatial confounder at one location affects the response at another location. Viewing spatial confounding as omission of spatial variable (causal inference definition), estimates will still be biased if only direct spatial confounding is accounted for and indirect spatial confounding is present.  We showed this theoretically that direct spatial confounding and indirect spatial confounding are also not the same just like spatial interference and direct spatial confounding are not the same. We evaluated this using simulation. In such case, the model (assuming no interference effects) is given as:  

\begin{equation}
 \textbf{M5} \quad \quad  \quad     Y\bf{(s_i)} = \beta_0 + \beta_a A\bf{(s_i)} + \beta_u U\bf{(s_i)} +  \beta_{\tilde{u}} \tilde{U}\bf{(s_i)}+ \epsilon_i 
    \label{equation 5}
\end{equation}

\noindent where $ \beta_{\tilde{u}}$ is the indirect spatial confounding effect, $\tilde{U} = \Phi U = \sum_{i \neq j} \phi_{ij} U_j $ is a function of some weights $\Phi \in \mathbb{R}^{n\times n }$ on spatial confounder $U$. Note that this is different from $\Psi$, the weights on the treatment $A$

We combined all the results and evaluate the bias with the full model i.e, model with treatment effect of interest, interference effect, direct spatial confounding and indirect spatial confounding. The full model is thus given as:
\begin{equation}
    \textbf{M6} \quad \quad Y\bf{(s_i)} = \beta_0 + \beta_a A\bf{(s_i)} + \beta_{\tilde{a}} \tilde{A}\bf{(s_i)} + \beta_u U\bf{(s_i)} +  \beta_{\tilde{u}} \tilde{U}\bf{(s_i)}+ \epsilon_i 
    \label{equation 6}
\end{equation}

\section{Bias derivation}
\subsection{Quantifying bias due to spatial interference}

As stated earlier, the bias when direct spatial confounding is present had been studied in literature (see \citet{paciorek2010}; \citet{page2017estimation}; \citet{dupont2023}; \citet{narcisi2024}; \citet{khan2023}) but study on quantifying and evaluating the bias in the presence of spatial interference and indirect spatial confounding are present in causal inference is limited. This is one of the gaps this study tends to fill. In this section, we derive analytical expression for the bias due to spatial interference for both spatial and non-spatial setting. We adopt the Least Squares method for the non-spatial case and the Generalised Least Squares approach for the spatial case. \\

\begin{proposition}
    Given a model of the form in Equation \ref{equation 2} and assuming only $\beta_a$ is unknown, the bias due to spatial interference in non-spatial and spatial settings are given respectively as $Bias\left(\hat{\beta_a^{*}}|A\right)=  \beta_{\tilde{a}}\frac{A^T \Psi A}{\left(A^{T} A\right)} $ and $Bias\left(\hat{\beta_a}|A\right) =  \beta_{\tilde{a}} \frac{A^T \Omega^{-1} \Psi A} {\left(A^{T} \Omega^{-1}A\right)}$.\\
\end{proposition}

\begin{proof}
   In non spatial setting, consider equation \ref{equation 2} given above as:

\begin{equation*}
    Y\bf{(s_i)} =  \beta_a A\bf{(s_i)} + \beta_{\tilde{a}} \tilde{A}\bf{(s_i)}+  \epsilon_i 
\end{equation*}
Note that the intercept has been excluded in the model assuming that the data is centered. Recall that Bias of $(\hat{\beta_a}) = E\left(\hat{\beta_a}\right) - \beta_a$ and the OLS estimator for $\beta_a$ is given as $\hat{\beta_a} = \left(A^TA\right)^{-1}A^TY$. 

Now; 
\begin{align*}
E\left(\hat{\beta_a^{*}}|A\right) &= \left(A^TA\right)^{-1}A^T\left( \beta_a A\bf{(s_i)} + \beta_{\tilde{a}} E\left(\tilde{A}\bf{(s_i)}|A \right)\right) \\
&= \beta_a \left(A^TA\right)^{-1}A^T A + \beta_{\tilde{a}} \left(A^TA\right)^{-1}A^T E\left(\tilde{A}|A \right) \\
&= \beta_a + \beta_{\tilde{a}} \left(A^{T} A\right)^{-1} A^T \tilde{A} \\
& = \beta_a + \beta_{\tilde{a}} \left(A^{T} A\right)^{-1} A^T \Psi A \\
 &= \beta_a +  \beta_{\tilde{a}}\frac{A^T  \Psi A}{A^{T} A}
\end{align*}
Hence, treatment effect bias due to spatial interference in non-spatial setting is given as:

\begin{equation}
  Bias\left(\hat{\beta_a^{*}}|A\right) =  \beta_{\tilde{a}}\frac{A^T \Psi A}{A^{T} A}  
  \label{SI non spatial}
\end{equation}

\noindent Similarly, in spatial setting, 
\begin{align*}
E\left(\hat{\beta_a}|A\right) &= \left(A^T \Omega^{-1}A\right)^{-1}A^T \Omega^{-1}\left( \beta_a A\bf{(s_i)} + \beta_{\tilde{a}} E\left(\tilde{A}\bf{(s_i)}|A \right)\right) \\
&= \beta_a \left(A^T \Omega^{-1}A\right)^{-1}A^T \Omega^{-1} A + \beta_{\tilde{a}} \left(A^T \Omega^{-1}A\right)^{-1}A^T \Omega^{-1} E \left(\tilde{A}|A \right) \\
&= \beta_a + \beta_{\tilde{a}} \left(A^{T}\Omega^{-1} A\right)^{-1} A^T \Omega^{-1}\tilde{A} \\
& = \beta_a + \beta_{\tilde{a}} \left(A^{T} \Omega^{-1} A\right)^{-1} A^T \Omega^{-1} \Psi A \\
&= \beta_a +  \beta_{\tilde{a}}\frac{A^T \Omega^{-1} \Psi A}{A^{T} \Omega^{-1}A} 
\end{align*}

\noindent Therefore, the bias due to spatial interference in spatial setting is given as:

\begin{equation}
Bias\left(\hat{\beta_a}|A\right) =    \beta_{\tilde{a}}\frac{A^T \Omega^{-1} \Psi A}{A^{T} \Omega^{-1}A} 
  \label{SI  spatial}
\end{equation}
\end{proof}


It can be seen that the bias induced by spatial interference is a function of the distributional properties of $A$, the strength of the interference effects and how the weight matrix, $\Psi$, has been chosen. Hence, in the simulation section, we evaluated the interference effects using two different weights, 
two distributions for the treatment (discrete and continuous) and two different cases for the interference effects (i.e. cases where the interference effect is lower and higher than the treatment effect).

\subsubsection{Choosing the weight matrix, \texorpdfstring{$\Psi$}{Psi}}
As shown and noted in the proceeding section that the bias induced by spatial interference depends heavily on the weight matrix, $\Psi$. Hence, how to choose the weight matrix plays a crucial role in determining the magnitude of the bias. Here, we briefly describe some ways how the weight matrix can be chosen in practice. 

The weight can be contiguity based, that is, chosen based on adjoining neighbours. This could be rook contiguity (when neighbours shared common edge) or queens contiguity (when neighbours are vertical or horizontal to one another). It could also be chosen based on the distance between locations. This distance can be fixed in advanced based on expert opinion, study requirement or convenience sake and units within this region are considered neighbours. The inverse distance can also be used, in this case, the weight decreases as the distance increases.  For example, suppose distance, $ d_{ij} = 20km$ is a pre-fixed distance, the weight between two units $i$ and $j$ can be $w_{ij} = \frac{1}{d_{ij}^p}$ where $p$ is a power parameter usually 1 or 2. A kernel or smooth decay function such as Gaussian or Epanechnikov can also be used in selecting the neighbours ($w_{ij} = e^{(d_{ij}^2/2)}$). This is flexible than the reciprocal of the distance. Also, irrespective of the distance between units, k-Nearest Neighbour (K-NN) can also be used in constructing the weights. Each units is associated with its k-nearest neighbour. In addition, the weight matrix can also be constructed using block based weights. This is useful when units are grouped into regions. The weights are then assigned based on within-group or between-group relationships.  Lastly, the weights can be constructed using criteria that are non-spatial such as shared features like language, or socio-economic status and so on. However, only the k-NN and distance based weights were evaluated in this study. 

\vspace{3mm}
\textbf{Constructing k-NN based weight matrix}

\begin{itemize}

        \item Step 1: For the choice of k (k=1,2,3,4,..), find the $k$ neighbours with the least distance for each of the unit $i$
        
        \item Step 2: For each of the units $i$, assigned weights to the identified $k$ neighbours as 
        
        $$w_{ij} = 
\begin{cases}
\dfrac{1}{k}, & \text{if } j \in \text{k-NN}(i) \\
0, & \text{otherwise}
\end{cases}
$$
  \end{itemize}
 

\vspace{4mm}

\textbf{Constructing the distance based weight matrix}

\begin{itemize}

\item Step 1: The full matrix with pairwise distance between units $i$ and $j$ is constructed. 
   
    \item Step 2: Determine or compute the percentile threshold or cutoff (e.g. 95\%, 50\%) of the distances.
   
    \item Step 3: For each pair $(i,j)$, compute the weights as:
    $$
w_{ij} = 
\begin{cases}
\dfrac{1}{d_{ij}}, & \text{if } d_{ij} \leq \text{threshold} \text{ and } where \;\;\; i \neq j \\
0, & \text{otherwise}
\end{cases}
$$ 

\end{itemize}

In both ways, the weights can be row standardised, such that $w_{ij} = \frac{w_{ij}}{\sum_k w_{ik}}.$ Each row of the matrix sums up to a unit for easy comparison of weights and to prevent overfitting.

\subsection{Direct spatial confounding and Bias with Poisson distributed treatment and confounder}

We are interested in spatial confounding as it relates to causal inference as defined in the introductory section. Having also established that the use of spatial random effects does not solve the problem of spatial confounding (\citet{schnell2020}; \citet{paciorek2010} ). Previous bias results (where both treatment and confounder are assumed to be normally distributed) obtained in literature (e.g. \citet{paciorek2010}; \citet{khan2023}) will change based on the distribution of the confounder and treatment. As an illustration, we propose a case where the unobserved confounder is Poisson distributed (for example, if frequency/number of visits to hospital in health studies is an unmeasured potential confounder) but with a spatial structure, introducing spatial confounding. We then derived the bias in the presence of direct spatial confounding when the exposure and missing confounder both have Poisson spatial structure. \\

\begin{proposition}
    Let the treatment \( A(\mathbf{s}_i) \) and confounder \( U(\mathbf{s}_i) \) be spatially distributed such that \( A(\mathbf{s}_i) = A_a(\mathbf{s}_i) + A_u(\mathbf{s}_i) \), where \( A_u(\mathbf{s}_i) \) is equivalent in distribution to  \( U(\mathbf{s}_i) \), and both components, $A_a(\mathbf{s}_i)$ and  $A_u(\mathbf{s}_i)$ are independent Poisson processes with rates \( \delta_1 (\mathbf{s}_i) \) and \( \delta_2 (\mathbf{s}_i) \), respectively. Then, the conditional distribution of \( U(\mathbf{s}_i) \) given \( A(\mathbf{s}_i) \) is \( \text{Binomial}(A, \frac{\delta_2}{\delta_1 + \delta_2}) \) assuming that $\delta_i$ is constant across each of the locations $\bf{s_i}.$
\end{proposition}

\begin{proof}

    Suppose that \( U(\mathbf{s}_i) \sim \text{Poisson}(\delta_2) \) and \( A_a(\mathbf{s}_i) \sim \text{Poisson}(\delta_1) \) are independent. Then their sum
\[
A(\mathbf{s}_i) = A_a(\mathbf{s}_i) + U(\mathbf{s}_i)
\]
is also Poisson distributed with rate \( \delta_1 + \delta_2 \), i.e.,
\[
A(\mathbf{s}_i) \sim \text{Poisson}(\delta_1 + \delta_2).
\]

Let us denote \( A = A(\mathbf{s}_i) \) and \( U = U(\mathbf{s}_i) \). The joint distribution of \( U \) and \( A \) can be expressed using the convolution of two independent Poisson variables:
\[
f(U = u, A = a) = f_U(u) \cdot f_{A_a}(a - u),
\]
where
\[
f_U(u) = \frac{\delta_2^u e^{-\delta_2}}{u!}, \quad f_{A_a}(a - u) = \frac{\delta_1^{a - u} e^{-\delta_1}}{(a - u)!}.
\]

The marginal distribution of \( A \) is:
\[
f_A(a) = \frac{(\delta_1 + \delta_2)^a e^{-(\delta_1 + \delta_2)}}{a!}.
\]

Therefore, the conditional distribution is:
\begin{align*}
f(U = u \mid A = a) &= \frac{f(U = u, A = a)}{f_A(a)} \\
&= \frac{\frac{\delta_2^u e^{-\delta_2}}{u!} \cdot \frac{\delta_1^{a - u} e^{-\delta_1}}{(a - u)!}}{\frac{(\delta_1 + \delta_2)^a e^{-(\delta_1 + \delta_2)}}{a!}} \\
&= \frac{\delta_2^u \delta_1^{a - u} e^{-(\delta_1 + \delta_2)} \cdot a!}{u!(a - u)! (\delta_1 + \delta_2)^a e^{-(\delta_1 + \delta_2)}} \\
&= \binom{a}{u} \frac{\delta_2^u \delta_1^{a - u}}{(\delta_1 + \delta_2)^a} \\
&= \binom{a}{u} \left( \frac{\delta_2}{\delta_1 + \delta_2} \right)^u \left( \frac{\delta_1}{\delta_1 + \delta_2} \right)^{a - u}.
\end{align*}

Thus, \( U \mid A \sim \text{Binomial}(A, \frac{\delta_2}{\delta_1 + \delta_2}) \).
\end{proof}

\begin{proposition} 
The bias induced by non normal spatial confounder (Poisson distributed) in estimating causal effects of spatial Poisson distributed exposure that varies at two scales with the confounder is $\beta_u \left[ \frac{\delta_2}{\delta_1 + \delta_2}\right ]$, where $\delta_1$, $\delta_2$ are the Poisson parameters and $\beta_u$ is the confounding effect. 
\end{proposition} 

\begin{proof}
Suppose $U\bf{(s_i)} \sim Poisson (\delta_2 \bf{(s_i)}) $ and $A\bf{(s_i)}$ varies at two scales where one is dependent on the confounder ($A_a\bf{(s_i)}$) and the other is independent on the confounder $A_u\bf{(s_i)}$ such that $A\bf{(s_i)} = A_a\bf{(s_i)} + A_u\bf{(s_i)}$.  If $ A_a\bf{(s_i)} \sim Poisson (\delta_1 \bf{(s_i)}) $ and  $A_u \bf{(s_i)} \sim Poisson (\delta_2 \bf{(s_i)})$ are independent, where $A\bf{(s_i)} \sim Poisson (\delta_1 \bf{(s_i)} \ + \delta_2 \bf{(s_i)} )$, then $U\bf{(s_i)}$ and  $A\bf{(s_i)}$ are conditionally independent. Now, using equation \ref{equation 4}, given as:

\begin{equation*}
    Y\bf{(s_i)} = \beta_a A\bf{(s_i)} + \beta_u U\bf{(s_i)} + \epsilon_i 
\end{equation*}

The bias in nonspatial settings in this case is derived as:
\begin{align*}
E\left(\hat{\beta_a^{*}}|A\right)_{C}^{Poi} &= \left(A^TA\right)^{-1}A^T\left( \beta_a A\bf{(s_i)} + \beta_u E\left( U\bf{(s_i)}|A\right) \right) \\
&= \beta_a + \beta_u \left(A^TA\right)^{-1}A^T  E\left( U\bf{(s_i)}|A\bf{(s_i)}\right)
\end{align*}
Using proposition 2, we have;
\begin{align*}
E\left(\hat{\beta_a^{*}}|A\right)_{C}^{Poi}  &= \beta_a + \beta_u \left(A^TA\right)^{-1}A^T \left[A \frac{\delta_2}{\delta_1 + \delta_2}\right ] \\
& = \beta_a + \beta_u \left(A^TA\right)^{-1}A^TA \left[ \frac{\delta_2}{\delta_1 + \delta_2}\right ] \\
&= \beta_a + \beta_u \left[ \frac{\delta_2}{\delta_1 + \delta_2}\right ]
\end{align*}

Hence, the bias in non-spatial setting is given as:

$$Bias\left(\hat{\beta_a^{*}}|A\right)_{C}^{Poi} = \beta_u \left[ \frac{\delta_2}{\delta_1 + \delta_2}\right ]$$

Similarly, the bias in spatial setting is derived as: 
\begin{align*}
E\left(\hat{\beta_a}|A\right)_{C}^{Poi} &= \left(A^T \Omega^{-1} A\right)^{-1}A^T \Omega^{-1} \left( \beta_a A\bf{(s_i)} + \beta_u E\left( U\bf{(s_i)}|A\right) \right) \\
&= \beta_a + \beta_u \left(A^T \Omega^{-1} A\right)^{-1}A^T \Omega^{-1} E\left( U\bf{(s_i)}|A\bf{(s_i)}\right)\\
&= \beta_a + \beta_u \left(A^T \Omega^{-1} A\right)^{-1}A^T \Omega^{-1} \left[A \frac{\delta_2}{\delta_1 + \delta_2}\right ] \\
& = \beta_a + \beta_u \left(A^T \Omega^{-1} A\right)^{-1}A^T \Omega^{-1} A \left[ \frac{\delta_2}{\delta_1 + \delta_2}\right ] 
\end{align*}

Hence, the bias in spatial setting is given as:

$$Bias\left(\hat{\beta_a}|A\right)_{C}^{Poi} = \beta_u \left[ \frac{\delta_2}{\delta_1 + \delta_2}\right ] $$
\end{proof}

Notice that the bias in both spatial and non-spatial settings when the treatment and confounder have a Poisson distribution is the same and independent of the variance-covariance matrix.

\subsection{Quantifying bias due to direct and indirect spatial confounding.}
The bias in the presence of direct spatial confounding, as shown in Equation \ref{equation 4}, has been derived and studied extensively (see \citet{paciorek2010} and \citet{khan2023}). Here we focus on deriving and evaluating the bias due to indirect spatial confounding. 

Consider equation \ref{equation 5} given as: 
\begin{align*}
    Y\bf{(s_i)} &=  \beta_a A\bf{(s_i)} + \beta_u U\bf{(s_i)} +  \beta_{\tilde{u}} \tilde{U}\bf{(s_i)}+ \epsilon_i 
\end{align*}

Assume that the treatment $A$ decomposes into two components: one correlated with the confounder and another independent Gaussian process. Specifically,
$$A = A_a + A_u, \quad \quad A \sim N(0, \sigma_a \Omega_a (\omega) + \sigma_u \Omega_u (\omega)),$$   $$U \sim N(0,\sigma_u \Omega_u (\omega)), \quad \quad  Cov (A, U) = \rho \sigma_a \sigma_u \Omega_u (\omega),$$ where $\Omega_u (\omega)$ and $\Omega_a (\omega)$ are spatial covariance matrices, $\rho$ is the correlation between A and U, and $\sigma_a$, and $\sigma_u$ are their respective marginal standard deviations.
\begin{align*}  
E\left(\hat{\beta_a^{*}}|A\right)_{SC} &= \left(A^TA\right)^{-1}A^T\left( \beta_a A\bf{(s_i)}+ \beta_u E\left( U\bf{(s_i)} |A\right) +\beta_{\tilde{u}} E\left(\tilde{U}\bf{(s_i)} | A\bf{(s_i)} \right) \right)\\
&= \beta_a + \beta_u \left(A^TA\right)^{-1}A^T E\left(U|A \right) +  \beta_{\tilde{u}} \left(A^TA\right)^{-1}A^T E\left(\tilde{U}|A \right) 
\end{align*}
The second part $\beta_u \left(A^TA\right)^{-1}A^T E\left(U|A \right) $ had been obtained  by \citet{paciorek2010} for non spatial setting and \citet{khan2023} for spatial setting (see appendix for details).

The main novelty here is with the third part which involves the indirect confounder $\beta_{\tilde{u}} \left(A^TA\right)^{-1}A^T E\left(\tilde{U}|A \right)$ i.e. bias due to indirect spatial confounding. \\

\begin{proposition}
The bias induced by an indirect spatial confounder $\tilde{U} = \Phi U$, where both $\tilde{U}$ and $A$ are Gaussian processes in non spatial and spatial settings, is given by:
$\frac{\beta_{\tilde{u}} M(\omega) A^T \Omega^{-1 } \phi^T  (A-\mu_a 1)}{A^T \Omega^{-1 }A}$ and $\beta_{\tilde{u}} M(\omega) ( A^T A)^{-1}A^T \phi^T  (A-\mu_a 1)$ respectively

where: $M(\omega) = \frac{\rho\sigma_a \sigma_u \Omega_u (\omega)} {\sigma_a^2 \Omega_a (\omega)+ \sigma_u^2 \Omega_u (\omega)}.$ Note that $A$ is centered.
\end{proposition} 

\begin{proof}
Starting with:
\begin{align*} 
\beta_{\tilde{u}} \left(A^TA\right)^{-1}A^T E\left(\tilde{U}|A \right)
&=\beta_{\tilde{u}} \left(A^TA\right)^{-1}A^T E\left(\Phi^T U|A\right)\\
&=\beta_{\tilde{u}} \left(A^TA\right)^{-1}A^T \Phi^T E\left(U| A \right) \\
&=\beta_{\tilde{u}} \left(A^TA\right)^{-1}A^T \Phi^T \left(  \mu_u \bf{1} + \frac{\rho\sigma_a \sigma_u \Omega_u (\omega)} {\sigma_a^2 \Omega_a (\omega)+ \sigma_u^2 \Omega_u (\omega)}   (A-\mu_a \bf{1} )  \right)
\end{align*}

Focusing on the second element, we have; 

\begin{align*}
\beta_{\tilde{u}} \left(A^TA\right)^{-1}A^T E\left(\tilde{U}|A \right) &=\beta_{\tilde{u}} \left(A^TA\right)^{-1}A^T \phi^T \left(\frac{\rho\sigma_a \sigma_u \Omega_u (\omega)} {\sigma_a^2 \Omega_a (\omega)+ \sigma_u^2 \Omega_u (\omega)}   (A-\mu_a 1)  \right)\\
&=\beta_{\tilde{u}} M(\omega) ( A^T A)^{-1}A^T \phi^T  (A-\mu_a 1)
\end{align*}

\noindent where $M(\omega) = \rho\sigma_a \sigma_u \Omega_u (\omega) [\sigma_a^2 \Omega_a (\omega)+ \sigma_u^2 \Omega_u (\omega)]^{-1}$ and all other notations take their usual meanings as previously defined. 

Therefore, bias induced by indirect spatial confounding in non spatial setting denoted by $Bias\left(\hat{\beta_a^{*}}|A\right)_{Ind. SC}$ is given as: 
\begin{align}  
Bias\left(\hat{\beta_a^{*}}|A\right)_{Ind. SC} &=\beta_{\tilde{u}} M(\omega) ( A^T A)^{-1}A^T \phi^T  (A-\mu_a 1)
\end{align}

Similarly, in spatial setting, bias induced by indirect spatial confounding denoted by $Bias\left(\hat{\beta_a^{}}|A\right)_{Ind. SC}$ is given as: 
\begin{align}  
Bias\left(\hat{\beta_a^{}}|A\right)_{Ind. SC} &= \frac{\beta_{\tilde{u}} M(\omega) A^T \Omega^{-1 } \phi^T  (A-\mu_a 1)}{A^T \Omega^{-1 }A}
\end{align}
\end{proof}
Hence, bias induced by direct and indirect spatial confounding in spatial and non spatial setting are given in equations \ref{equation direct and indirect sc non spatial} and \ref{equation direct and indirect sc spatial} respectively as: 
\begin{equation}  
Bias\left(\hat{\beta_a^{*}}|A\right)_{direct + ind. SC} = \beta_u \rho \frac{\sigma_u}{\sigma_a}\left[ \left(A^{*T}A^*\right)^{-1} A^{*T}K \left(A - \mu_a1 \right)  \right]_2 + \beta_{\tilde{u}} M(\omega) ( A^T A)^{-1}A^T \phi^T  (A-\mu_a 1)
\label{equation direct and indirect sc non spatial}
\end{equation}
and 
\begin{align}  
Bias\left(\hat{\beta_a^{}}|A\right)_{direct + ind. SC} &= \beta_u \rho \frac{\sigma_u}{\sigma_a}\left[ \left(A^{*T} \Omega^{-1}A^*\right)^{-1} A^{*T} \Omega^{-1} K \left(A - \mu_a1 \right)  \right]_2 \\
&  +\frac{\beta_{\tilde{u}} M(\omega) A^T \Omega^{-1 } \phi^T  (A-\mu_a 1)}{A^T \Omega^{-1 }A}
\label{equation direct and indirect sc spatial}
\end{align}

\subsection{Bias in the presence of spatial interference and spatial confounding (direct and indirect)}
In this section, the assumption of no confounding and no interference are relaxed and the bias due to both phenomenon were derived. 


Using equation \ref{equation 6}, model with spatial interference, direct and spatial confounding given as:
\begin{equation*}
     Y\bf{(s_i)} = \beta_a A\bf{(s_i)} + \beta_{\tilde{a}} \tilde{A}\bf{(s_i)} + \beta_u U\bf{(s_i)} +  \beta_{\tilde{u}} \tilde{U}\bf{(s_i)}+ \epsilon_i 
\end{equation*}

\noindent and following our results of bias derived in equations \ref{SI non spatial} and \ref{equation direct and indirect sc non spatial} the bias due to spatial interference and spatial confounding in non spatial setting for normally distributed confounder and treatment is given as follow:
\begin{align*}
 Bias\left(\hat{\beta_a^{*}}\right)_{SI\_SC} &= \beta_{\tilde{a}} \left(A^{T} A\right )^{-1} A^T \Psi A  + \beta_u \rho \frac{\sigma_u}{\sigma_a}\left[ \left(A^{*T}A^*\right)^{-1} A^{*T}K \left(A - \mu_a1 \right)  \right]_2 \\
 & +\beta_{\tilde{u}} M(\omega) ( A^T A)^{-1}A^T \phi^T  (A-\mu_a 1)
 \end{align*}


Similarly, in spatial setting, using results in equations \ref{SI  spatial} and \ref{equation direct and indirect sc spatial}, the bias dues to spatial interference and spatial confounding in spatial setting is given as:
\begin{align*}
 Bias\left(\hat{\beta_a} \right)_{IC} &= \beta_{\tilde{a}} \left(A^{T} \Omega^{-1}A\right )^{-1} A^T \Psi A  + \beta_u \rho \frac{\sigma_u}{\sigma_a}\left[ \left(A^{*T} \Omega^{-1}A^*\right)^{-1} A^{*T} \Omega^{-1} K \left(A - \mu_a1 \right)  \right]_2 \\
 & + \beta_{\tilde{u}} M(\omega) \left( A^T \Omega^{-1 }A \right)^{-1} A^T \Omega^{-1 } \phi^T  (A-\mu_a 1)
 \end{align*}
 
Note that we have centered the treatment variable by leaving out the intercept $\beta_0$ in obtaining the bias due to interference effect and indirect spatial confounding, but leave the variable uncentered for the bias due to direct spatial confounding just to replicate what \citet{khan2023} obtained in their study.

\section{Simulation study}
In this section, we present a series of simulation experiments to investigate the bias introduced in estimating the treatment effect under various spatial data-generating mechanisms. The simulation study is structured around four key objectives:

\begin{enumerate}
\item To assess the bias introduced by spatial interference, with attention to the impact of different weight structures, magnitudes of interference effects, and treatment distributions.
\item To examine how the distributional assumptions of both the treatment and confounder affect the estimation of bias.
\item To explore whether direct and indirect spatial confounding contribute similarly or differently to bias.
\item To quantify the bias when any of the three spatial phenomena  - interference, direct, or indirect confounding - are present but unaccounted for in the analysis.
\end{enumerate}

\subsection{Evaluation metrics}
Through out this paper, the root mean square error and the coverage probability were metrics employed for evaluating the simulation estimates obtained due to all the phenomenon.  

\begin{itemize}
    \item \textbf{Root Mean Square Error (RMSE):} The root mean square error is a metric for evaluating how large the estimated values is from the true parameter value on average. Though sensitive to outliers but it gives higher weights to large errors and is easily interpretable. A lower RMSE signifies a better model fit or better estimate. 

    RMSE for $\hat{\beta}$ is given by:

\begin{equation*}
    RMSE \; (\hat{\beta})  = \sqrt{Var(\hat{\beta)} + Bias(\hat{\beta})^2} 
\end{equation*}

    where Bias $(\hat{\beta}) = E(\hat{\beta}) - \beta$ and Var $(\hat{\beta}) =E\left[(\hat{\beta} - E(\hat{\beta}))^2 \right]$
    
Using simulation,  
    \begin{equation*}
       \hat{ RMSE \; (\hat{\beta}) } = \sqrt{\frac{1}{N-1} \sum_{i=1}^{N} \left(\hat{\beta}_i - \bar{\hat{\beta} }\right)^2 + \left({ \bar{\hat{\beta} } - \beta} \right)^2 } 
    \end{equation*}

where $$ \bar{\hat{\beta} } = E(\hat{\beta}) = \frac{1}{N} \sum_{i=1}^{N} \hat{\beta}_i$$
$N, \beta $ and $\hat{\beta_i}$ are the number of simulations, actual treatment effect and the estimated treatment effect for each $i$-th simulation respectively. If the estimator is unbiased, then RMSE is just equal to the square root of the variance.

    \item \textbf{Coverage Probability:} The coverage probability is the number of times or simply the proportion of times that the true parameter value falls within the confidence interval when carried out repeatedly over several samples or trials.

This is given as:

\begin{equation*}
    Coverage \; Probability = \frac{1}{n} \sum_{i=1}^n  \mathds{1} (\beta \in  C_i)
\end{equation*}
where $\mathds{1}$ is an indicator function which is 1 if the parameter $\beta$, truly falls within the confidence interval $C_i$ and $O$ if not while $n$ is the number of confidence interval. It is usually expressed in percentage. As illustration, if a simulation is performed 1000 times and for each of the simulations, 90\% confidence interval is constructed for the estimated parameter. A 90\% coverage probability means that the true parameter is contained 900 times in the interval out of the 1000 confidence intervals constructed. In cases where coverage probability is higher than the confidence interval, the interval is referred to as conservative and permissive if lower. 
    
\end{itemize}

\subsection{Bias induced by spatial interference}

We begin by simulating data from the model in Equation~\ref{equation 2}:
\begin{equation*}
Y(\mathbf{s}_i) = \beta_a A(\mathbf{s}_i) + \beta_{\tilde{a}} \tilde{A}(\mathbf{s}_i) + \epsilon_i
\end{equation*}

Here, $A(\mathbf{s}_i)$ is the spatially varying treatment, and $\tilde{A}(\mathbf{s}_i)$ captures the spatial interference. We simulate  under both non-spatial and spatial scenarios. In the non-spatial case, $\epsilon_i \sim N(0,1)$. In the spatial case, $\epsilon_i$  follows a zero-mean Gaussian process with an exponential covariance structure and range parameter of 2.

We randomly sample 100 spatial locations within the grid $\left[0, 10\right]  \times \left[0, 10\right]$ and construct $A(\mathbf{s}_i)$ from a Gaussian process with zero mean and an exponential covariance function with range 1.

We consider four scenarios:
\begin{itemize}
\item \textbf{Scenario 1:} Data generated with interference, model fitted without interference.
\item \textbf{Scenario 2:} Data generated with interference, model fitted accounting for interference.
\item \textbf{Scenario 3:} Data generated without interference, model fitted without interference.
\item \textbf{Scenario 4:} Data generated without interference, model fitted with interference.
\end{itemize}

For each scenario, we evaluate two cases:
\begin{itemize}
\item \textbf{Case 1:} Treatment effect is stronger than the interference effect ($\beta_a > \beta_{\hat{a}}$).
\item \textbf{Case 2:} Treatment effect is weaker than the interference effect ($\beta_a < \beta_{\hat{a}}$).
\end{itemize}

The interference structure  is generated using two weight types: k-nearest neighbors ($k=4$) and distance-based weights (using the 95th percentile as the threshold). Each simulation setting is repeated over 1,000 iterations.

\subsubsection{Simulation results for non-spatial setting} 

The results for discrete and continuous treatments in non-spatial setting are reported in Table~\ref{Discrete and continuous non spatial}. Corresponding boxplots with uncertainties appear in Figure~\ref{k-NN and distance, discrete and continuous}.

\begin{table}[H]
\centering
\caption{Simulation result for discrete and continuous treatment in non spatial setting }
\label{Discrete and continuous non spatial}
\begin{small}
\begin{tabular}{ll|ccc|ccc}
\hline \hline
\multicolumn{2}{c|}{Discrete treatment}  &   \multicolumn{3}{|c|}{k-NN based weight (k=4)}  &  \multicolumn{3}{|c} {Distance based weight - 95\% threshold} \\
\hline \hline
\textbf{Cases} & \textbf{Scenario} & \textbf{Bias} & \textbf{RMSE} & \textbf{Coverage} & \textbf{Bias} & \textbf{RMSE} & \textbf{Coverage}\\
\hline \hline
\multirow{2}{*}{$\beta_a = 8, \beta_{\hat{a}} = 2$} 
  & Scenario 1 & 0.5054 & 0.5163 &    0.436 & -0.0125 &  0.157  &  0.946\\
  & Scenario 2 & -0.0076 & 0.1702 &   0.953 &0.0054 & 0.1677 &   0.949 \\
\multirow{2}{*}{$(\beta_a > \beta_{\hat{a}})$} 
  & Scenario 3 & -0.0050 & 0.1530   & 0.954 &0.0070 &0.1549 &   0.951  \\
  & Scenario 4 & -0.0076 & 0.1702  &  0.953 &-0.0054 & 0.1677  &  0.949 \\
\hline \hline
\multirow{2}{*}{$\beta_a = 3, \beta_{\hat{a}} = 9$} 
     & Scenario 1 & 2.2921 & 2.3052 &    0.117  & -0.0808 & 0.2073  &  0.871  \\
  & Scenario 2 &  -0.0076 & 0.1702 &   0.953 &0.0054 & 0.1677 &   0.949 \\ 
  \multirow{2}{*}{$\beta_a < \beta_{\hat{a}}$} 
   & Scenario 3 & -0.0050 & 0.1530   & 0.954 & 0.0070 &0.1549 &   0.951\\
  & Scenario 4 & -0.0076 & 0.1702  &  0.953 &-0.0054 & 0.1677  &  0.949 \\
\hline \hline

\multicolumn{2}{c|}{Continuous treatment}   &   \multicolumn{3}{|c|}{k-NN based weight (k=4)}  &  \multicolumn{3}{|c} {Distance based weight - 95\% threshold} \\
\hline \hline
\multirow{2}{*}{$\beta_a = 8, \beta_{\hat{a}} = 2$} 
  & Scenario 1 & 0.7338 & 1.0470 &   0.894 & 0.1568 & 3.6976 &   0.952\\
  & Scenario 2 & -0.0389 & 0.9834 &    0.957 &0.0902 & 4.1877  &  0.950 \\
\multirow{2}{*}{$(\beta_a > \beta_{\hat{a}})$} 
  & Scenario 3 &-0.0118 & 0.8086 &    0.951 & 0.1651 & 3.6958 &   0.953\\
  & Scenario 4 &-0.0389 & 0.9834 &   0.957 & 0.0902 & 4.1877 &   0.950\\
\hline \hline
\multirow{2}{*}{$\beta_a = 3, \beta_{\hat{a}} = 9$} 
     & Scenario 1 & 3.3435 & 3.37432  &  0.270  & 0.1278 & 3.7109 &    0.951 \\
  & Scenario 2 & -0.0389 & 0.9834 &    0.957 & 0.0902 & 4.1877 &   0.950 \\ 
  \multirow{2}{*}{$\beta_a < \beta_{\hat{a}}$} 
   & Scenario 3 & -0.0118 & 0.8086 &    0.951 & 0.1651 & 3.6958 &    0.953 \\
  & Scenario 4 & -0.0389 & 0.9834  &    0.957 &0.0902 & 4.1877 &    0.950  \\
\hline \hline
\end{tabular}
\end{small}
\end{table}

Table~\ref{Discrete and continuous non spatial} shows the simulation results of the bias evaluation induced by spatial interference in non spatial setting for discrete and continuous treatments. The result showed that the average bias is highest in scenario 1 (when interference is present but not accounted for) with highest RMSE.  This shows that interference effect when present and neglected affects the inference on the estimates. However, this bias is sensitive to the choice of weight used as the bias is higher when weights are selected using k nearest neighbours than when the distance based weight is used.  Also, the bias is higher in most cases when the interference effect is higher than the treatment effect for both weights selection methods. The result revealed that both continuous and discrete treatment bias behave in similar pattern with respect to the weights methods and effect of interference effects size examined.  A key distinguish feature is that, the bias is higher when the treatment is continuous than when it is discrete. 

\begin{figure}[H]
\centering
\includegraphics[width=0.8\textwidth]{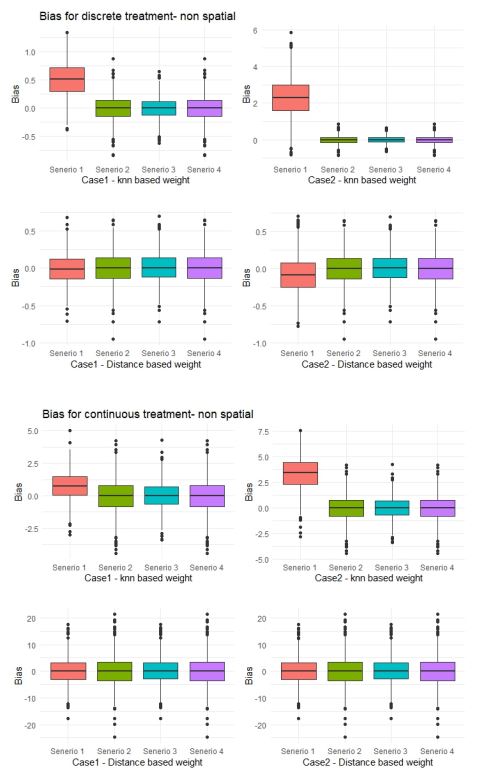}
    \caption{ Box plots showing the bias for discrete and continuous treatment with k-NN and distance based weights in non spatial setting}
    \label{k-NN and distance, discrete and continuous}
\end{figure}

Figure \ref{k-NN and distance, discrete and continuous} shows the box plots for discrete and continuous treatments for the cases considered and the two weights used. The figure revealed that the magnitude of the bias when interference is present and not accounted for is highest in all cases examined and weights used. In fact, the figure revealed that the range of bias in this situation can be negative or positive which can alter causal interpretation. The bias range is widest when the interference effect is higher that the treatment effect (case 2) and k-NN based weight is used. However, the bias is larger in continuous treatment when compared with the discrete treatment.

\subsubsection{Simulation results for spatial setting}
In the spatial setting, we retain the same scenarios and cases as before, but now model $\epsilon$ as spatially correlated noise. The results for discrete and continuous treatments are shown in Table~\ref{Discrete and Continuous Treatment Spatial}. Corresponding boxplots showing the uncertainty are provided in Figure~\ref{k-NN and distance based, discrete and Continuous Spatial}.

\begin{table}[H]
\centering
\caption{Simulation result for discrete and continuous treatment in spatial setting }
\label{Discrete and Continuous Treatment Spatial}
\begin{small}
\begin{tabular}{ll|ccc|ccc}
\hline \hline
Discrete &  &   \multicolumn{3}{|c|}{k-NN based weight (k=4)}  &  \multicolumn{3}{|c} {Distance based weight- 95\% threshold} \\
\hline \hline
\textbf{Cases} & \textbf{Scenario} & \textbf{Bias} & \textbf{RMSE} & \textbf{Coverage} & \textbf{Bias} & \textbf{RMSE} & \textbf{Coverage}\\
\hline \hline
\multirow{2}{*}{$\beta_a = 8, \beta_{\hat{a}} = 2$} 
  & Scenario 1 & -0.2550 &  0.284  &  0.630 & -0.0161 & 0.128  &  0.945\\
  & Scenario 2 & 0.0080 & 0.134   & 0.939 & 0.0057 & 0.144 &    0.938\\
\multirow{2}{*}{$(\beta_a > \beta_{\hat{a}})$} 
  & Scenario 3 & 0.0044 & 0.127  &  0.943 & 0.0044 & 0.127 &    0.943\\
  & Scenario 4 & 0.0080 & 0.134  &  0.939 & 0.0057 & 0.144 &   0.938\\
\hline \hline
\multirow{2}{*}{$\beta_a = 3, \beta_{\hat{a}} = 9$} 
     & Scenario 1 & -1.3700  &  1.380  &  0.140   & -0.0881 & 0.148  &  0.912 \\
  & Scenario 2 & 0.0080 & 0.134 &  0.939 & 0.0057 & 0.144 &    0.938\\ 
  \multirow{2}{*}{$\beta_a < \beta_{\hat{a}}$} 
   & Scenario 3 & 0.0044 &0.127 &   0.943 & 0.0044 & 0.127 &  0.943  \\
  & Scenario 4 & 0.0080 & 0.134  &  0.939 & 0.0057 & 0.144 &    0.938  \\
\hline \hline
Continuous &  &   \multicolumn{3}{|c|}{k-NN based weight (k=4)}  &  \multicolumn{3}{|c} {Distance based weight - 95\% threshold} \\
\hline \hline
\multirow{2}{*}{$\beta_a = 8, \beta_{\hat{a}} = 2$} 
  & Scenario 1 & -0.0887 &  0.1610   &  0.830 & -0.0141 & 0.0832  &  0.936\\
  & Scenario 2 &  0.00451 & 0.0829 &   0.942& 0.0063 & 0.0887  &  0.933\\
\multirow{2}{*}{$(\beta_a > \beta_{\hat{a}})$} 
  & Scenario 3 & 0.00540 & 0.0827 &   0.936 & 0.0054 & 0.0827  &  0.936\\
  & Scenario 4 & 0.00451 & 0.0829 &    0.942 & 0.0063 & 0.0887 &   0.933\\
\hline \hline
\multirow{2}{*}{$\beta_a = 3, \beta_{\hat{a}} = 9$} 
     & Scenario 1 & -0.5720 &  0.6820 &    0.687  & -0.0859  & 0.113 &    0.854 \\
  & Scenario 2 & 0.0045 & 0.0829  &  0.942 & 0.0063 & 0.0887 &   0.933\\ 
  \multirow{2}{*}{$\beta_a < \beta_{\hat{a}}$} 
   & Scenario 3 &  0.00540 & 0.0827 &   0.936 & 0.0054 & 0.0827  &  0.936  \\
  & Scenario 4 & 0.00451 & 0.0829 &   0.942 & 0.0063 & 0.0887 &   0.933  \\
\hline \hline
\end{tabular}
\end{small}
\end{table}

Similarly, in spatial setting as shown in Table~\ref{Discrete and Continuous Treatment Spatial} for discrete and continuous treatments. The bias is highest in scenario 1 and also with high RMSE. In contrary, the treatment effect is biased downward and also higher when the distance based weight is used compared to when k-NN based weight is used as against what was obtained in non spatial setting. In addition, the bias is higher in case 2 (situations where interference effect is higher than the treatment effect) for discrete treatment.

\begin{figure}[H]
\centering
\includegraphics[width=0.8\textwidth]{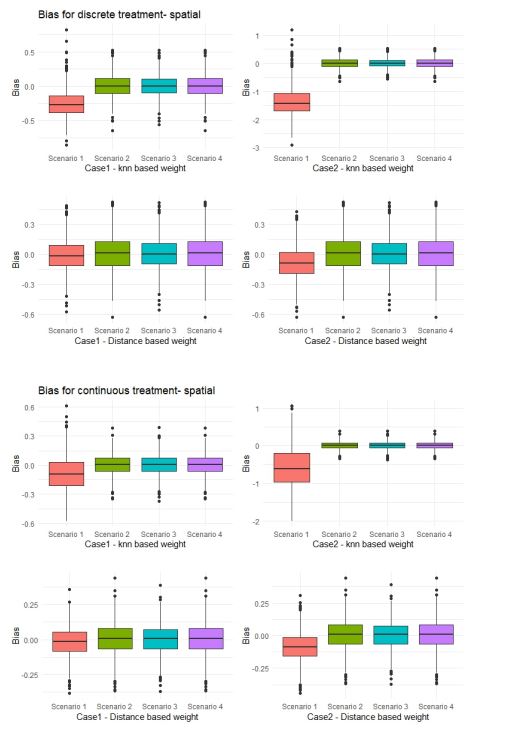}
   \caption{Box plots showing the bias for discrete and continuous treatment with k-NN and distance based weights in spatial setting}
  \label{k-NN and distance based, discrete and Continuous Spatial} 
\end{figure}

Figure~\ref{k-NN and distance based, discrete and Continuous Spatial} shows the box plot for discrete and continuous  treatment for the cases considered in spatial setting. For the discrete treatment, the result shows similar patterns to what was obtained in non spatial setting but the range of bias is lower in this case due to the structure of the spatial error.  The range of bias is highest in scenario 1, case 2.  Also, 
for continuous treatment, the bias is highest in scenario 1 but in contrary, the treatment effect is higher when the distance based weight is used compared to when k-NN is used as against what was obtained in non spatial setting. Surprisingly, the bias is lower when the treatment is continuous than when the treatment is discrete in spatial setting.

\subsection{Bias induced by spatial confounder with different distributions of treatment and confounder.}

Here, we examined the bias induced by spatial confounder when the distributions of the treatment and confounder changes. The goal is to investigate how the bias due to spatial confounding changes when both the treatment and confounder are not normally distributed.  We simulated a model of the form in equations \ref{equation 4} given as:
\begin{equation*}
    Y\bf{(s_i)} =  2 A\bf{(s_i)} + 1.5 U\bf{(s_i)} + \epsilon_i 
\end{equation*}
\noindent where $A\bf{(s_i)}$, $U\bf{(s_i)}$ are spatial treatments, spatial confounders and spatial random errors respectively. $A\bf{(s_i)}$ and $U\bf{(s_i)}$ are correlated and  we placed causal assumption from $U\bf{(s_i)}$ on $A\bf{(s_i)}$. 

We randomly sample 100 spatial locations within the grid [0, 10]×[0, 10] and construct $A\bf{(s_i)}$ and $U\bf{(s_i)}$ from a Gaussian process with zero mean and an exponential covariance function with range parameter of 2. In our simulation, we varied the distributions of both the treatment and confounders. Specifically, we simulated both to be normal, Poisson and binary but the $\epsilon_i$ is only spatial. We then fitted the model without the confounder i.e. $Y\bf{(s_i)} =  \beta_a A\bf{(s_i)} + \epsilon_i $ to examine the bias induced on the treatment effect for each types of distributions. 
Each simulation is repeated over 10,000 iterations for precision. The simulation results are presented in table \ref{different distributions of treatment and confounder} and shown in figure \ref{Treatment and Confounders Bias with RMSE}. We show the simulated Normal, Poisson and  Binary spatial treatment and spatial confounders with corresponding outcomes pictorially in figure \ref{Confounder and treaments}.  

\subsubsection{Simulation results}

The plot or representation of the simulated confounder and treatment for normal, Poisson and binomial distributions with their corresponding outcomes is shown in Appendix \ref{app: B3}

\begin{table}[H]
\centering
\caption{Treatment effect bias for different distributions of treatment and confounder}
\label{different distributions of treatment and confounder}
\begin{tabular}{|c|c|c|c|}
\hline \hline
\textbf{Confounder and Treatment} & \textbf{Bias} & \textbf{RMSE} & \textbf{Coverage} \\
\hline \hline
 Normal, Normal  & 0.0049 &  0.1732 &  0.93\\
 \hline
  Binary, Binary &  -0.0005 &  0.2288 & 0.894 \\
  \hline
  Poisson, Poisson & 0.0038 & 0.1796 & 0.93  \\
\hline \hline
\end{tabular}
\end{table}

In furtherance, comparing the bias for different distributions for the confounder and treatment, the result shown in Table~\ref{different distributions of treatment and confounder} revealed that the bias are not the same when the distributions of both the confounder and the treatment change. In fact, the bias could alter causal interpretation as it could be negative (binary case) or positive (normal and Poisson cases). The bias is highest in normal spatial treatment and normal spatial confounder with lowest RMSE. Followed by in Poisson spatial treatment and confounder. The bias is lowest when the confounder and treatment are both binary distributed with lowest coverage probability. The key finding here is that the bias changes significantly when the distribution of the treatment and confounders change. Strengthening our argument and motivation for deriving the bias when both the confounder and treatment are not normally distributed. The bias box plots for binary treatment and confounder, normal treatment and confounder and Poisson treatment and confounder is shown in Appendix \ref{app: B4}.

\subsection{Bias in the presence of spatial interference, direct and indirect spatial confounding}

In this subsection, the goal is to verify the bias when indirect spatial confounding or any other effect when present is ignored in causal inference. In particular, to investigate whether direct or indirect spatial confounding are the same or not as previously assumed in causal inference to support our theoretical derivations. We simulated a full model i.e model with treatment, spatial interference, direct and indirect spatial confounding ( T + I + DSC + ISC) of the form in Equation \ref{equation 6}. The simulated model is thus given as:

\begin{equation*}
    Y\bf{(s_i)} = 5 A\bf{(s_i)} + 3 \tilde{A}\bf{(s_i)} + 2.5 U\bf{(s_i)} +  2\tilde{U}\bf{(s_i)}+ \epsilon_i 
\end{equation*}

Nested models including models with; treatment and spatial interference(T + I), Treatment and direct spatial confounding (T + DSC), treatment and indirect spatial confounding (T + ISC), treatment, interference and direct spatial confounding (T + I + DSC), treatment, interference and indirect spatial confounding (T + I + ISC), treatment,  direct and indirect spatial confounding (T + DSC + DSC) were then fitted to examine the bias in the absence of each of them. 

We randomly sample 100 spatial locations within the grid [0, 10]×[0, 10] and construct $A\bf{(s_i)}$ and $U\bf{(s_i)}$ from a Gaussian process with zero mean and an exponential covariance function with range parameter of 2.
For the interference effect, k-NN and distance based weights are also examined, $\epsilon_i$ was only spatial here and simulated as previous discussed in section 5.1. The estimation was done using gls function in nlme package in R.

\subsubsection{Simulation results}

\begin{table}[H]
\centering
\caption{Bias in the presence of spatial interference, direct and indirect spatial confounding with k-NN and distance based weights}
\label{spatial interference, direct and indirect spatial confounding - k-NN and distance weight}
\begin{tabular}{|c|c|c|c|c|c|c|c|}
\hline \hline
 \textbf{S/N} & \textbf{Models} & \multicolumn{3}{c|}{\textbf{k-NN based weight (k=4)}} & \multicolumn{3}{c|}{\textbf{Distance based weight (95\%)}} \\
 \hline \hline
 &  & \textbf{Bias} & \textbf{RMSE} & \textbf{Coverage} & \textbf{Bias} & \textbf{RMSE} & \textbf{Coverage} \\
\hline \hline
1  &  T + I  & -0.0244  & 0.3980  &    0.97 & -0.0125 & 0.2530  &    0.93\\
\hline
 2 &     T + DSC & 0.5890 &  0.7260  &   0.54 & -0.0078 & 0.0913   &  0.91\\
 \hline
 3  &   T + ISC  & 1.1900  &  1.2100  &    0.23 & -0.0166 &  0.2450    &  0.93
\\
\hline
4  &    T + I + DSC & 0.0077 & 0.1690  &  0.94  & 0.0151  & 0.0885   &  0.96 \\
\hline
5  &  T + I + ISC  & 0.0018 & 0.1990  &   0.98 &  -0.0113 & 0.2400  &    0.93\\
\hline
6   &  T + DSC + ISC & -0.2560&   0.4240   &  0.45 & -0.0061 & 0.0872   &  0.92  \\
\hline
7  & T + I + DSC + ISC & 0.0218 & 0.0815  &  0.93 &    0.0177 &  0.0832  &   0.96  \\
\hline  \hline
\end{tabular}
\end{table}

Table~\ref{spatial interference, direct and indirect spatial confounding - k-NN and distance weight} shows the bias in the presence of spatial interference, direct spatial confounding and indirect spatial confounding for normally distributed treatment and confounder using k-NN and distance based weights. The result revealed that the bias is higher in models where interference effect is neglected with higher RMSE irrespective of the weighting methods. This is corroborating the importance of not neglecting interference effects in spatial causal inference when present. 
\begin{figure}[H]
  \centering
\includegraphics[width=0.8\textwidth]{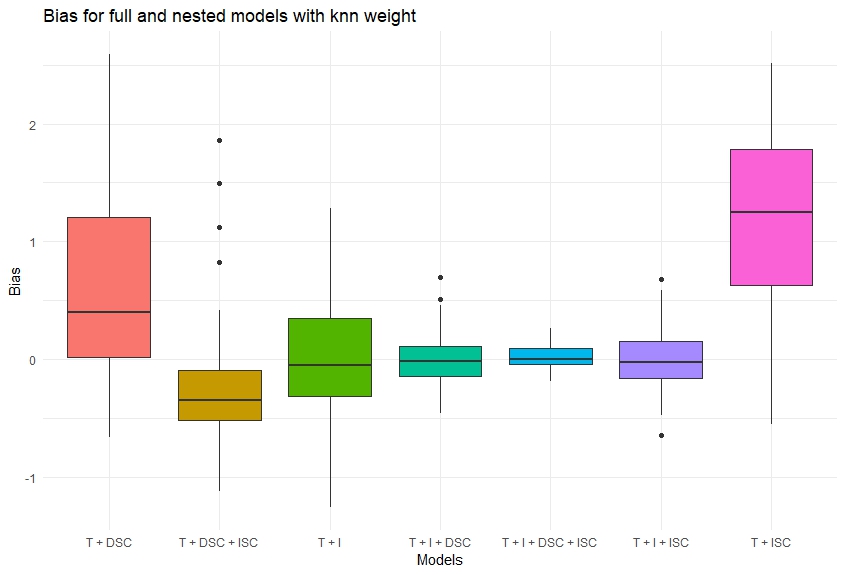}
   \caption{Box plots showing the bias for full and nested models with k-NN weights}
  \label{k-NN + I + DSC + ISC} 
\end{figure}

\begin{figure}[H]
  \centering
\includegraphics[width=.8\textwidth]{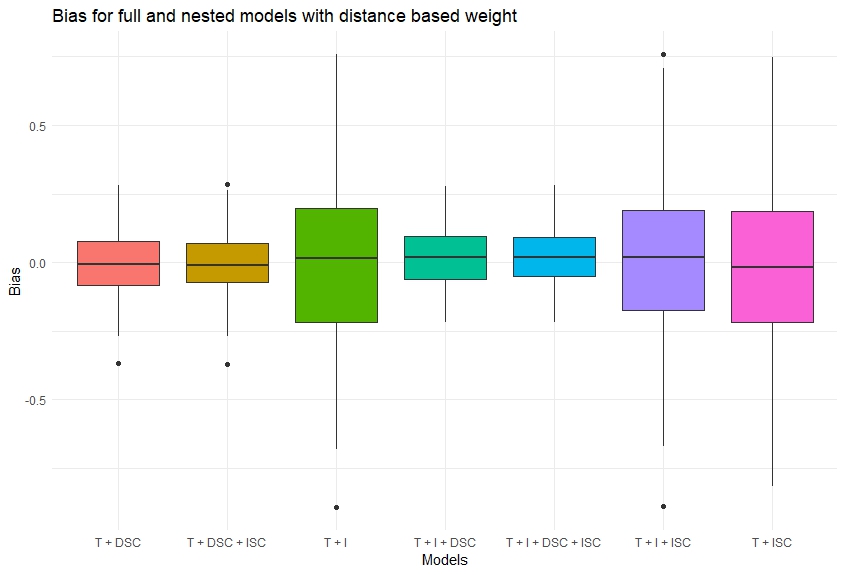}
   \caption{Box plots showing the bias for full and nested models with distance based weights}
  \label{Distance I + DSC + ISC} 
\end{figure}

Figures \ref{k-NN + I + DSC + ISC} and  \ref{Distance I + DSC + ISC} show the bias for the full model (model with treatment, interference, direct spatial confounding and indirect spatial confounding) and nested models. Meanwhile, using k-NN weight based, the result showed that when interference effect and direct spatial confounding only is accounted for, the range of treatment estimate bias is more  than when indirect spatial confounding is accounted for as evident in Figure \ref{k-NN + I + DSC + ISC}. But in contrary using distance based weights, the bias are almost similar and falls within the same range as shown in Figure \ref{Distance I + DSC + ISC}. Hence, we conclude that indirect spatial confounding can not always be neglected but dependent on how the confounders interferes with responses and treatments in another location.

\section{Application}

As illustration and in order to  disentangle spatial interference, direct and indirect spatial confounding biases, we provided application using Manchester datasets on air temperature, land surface temperature and precipitation between the period of 1st - 31st of January 2018 from Manchester Climate and Surface Daily Panel (2018). The data is an open access obtained from Google Earth Engine \citep{C3S} and available by \citep{uñoz} 

Air temperature being the degree of hotness or coldness of the air is one of the mostly measured environmental factor which affects other climate variables such as Land surface temperature.  As a result, air temperature is treated as the treatment and the land surface temperature is the outcome studied here. In addition, precipitation/rainfall could be a possible confounder which affects both air and land surface temperatures. For example, rainfall could cause the air temperature to be cooler due to increase in cloud cover or humidity. Also, after rainfall, soil wetness and moisture tend to increase and invariably affects land surface temperature. Hence, precipitation is taken as confounder. The air temperature is measured at 2 meters above the ground, Land surface temperature is measured during the night for robustness while precipitation is measured during the day

We fitted the models T + I, T + DSC, T + ISC, T + I + DSC, T + I + ISC, T + DSC + ISC and T + I + DSC + ISC as done in Table \ref{spatial interference, direct and indirect spatial confounding - k-NN and distance weight} to examine the bias when spatial interference and spatial confounder are both present but not accounted for and to distinguish direct and indirect spatial confounding. The results are shown in Table ~\ref{Bias results for the application data}. K-NN and distance based weights methods were both used in constructing the weights for the interference effects. Attention was paid mainly to the treatment effect estimate. The plots showing the treatment estimates and the confidence intervals are presented in Figure \ref{app - data}. The models were fitted using likfit() in geoR package \citep{ribeiro2020package}

\begin{table}[H]
\centering
\caption{Treatment effect estimates for the application data}
\label{Bias results for the application data}
\begin{tabular}{|c|c|c|c|c|c|c|c|}
\hline \hline
 \textbf{S/N} & \textbf{Models} & \multicolumn{2}{c|}{\textbf{k-NN based weight (k=4)}} & \multicolumn{2}{c|}{\textbf{Distance based weight (50\%)}} \\
 \hline \hline
 & . & \textbf{Estimate} & \textbf{Conf. Interval} & \textbf{Estimate} & \textbf{Conf. Interval.} \\
\hline \hline
1  &  T + I  &  0.1786 & [0.1680, 0.1891]  & 0.1426 &  [0.1396 , 0.1459]    \\
\hline
 2 &     T + DSC & -0.2068 & [-0.2104, -0.2031] & -0.2068 & [-0.2104, -0.2031]  \\
 \hline
 3  &   T + ISC  & -0.2045 & [-0.2081, -0.2008] & -0.2066 & [-0.2102,  -0.2029] \\
\hline
4  &    T + I + DSC & -0.1667 & [-0.1764, -0.1569]& -0.1817 & [-0.1853, -0.1782] \\
\hline
5  &  T + I + ISC  & 0.1890 & [0.1601, 0.2178]& 0.4012 & [0.3941, 0.4083] \\
\hline
6   &  T + DSC + ISC & -0.2070 &[-0.2106, -0.2033]& -0.2072 & [-0.2108, -0.2035] \\
\hline
7  & T + I + DSC + ISC & -0.1898 & [-0.2020, -0.1776] & -0.2037 & [-0.2074,  -0.2000]\\
\hline  \hline
\end{tabular}
\end{table}

Table \ref{Bias results for the application data} shows the treatment effect estimate in the presence of spatial interference, direct spatial confounding and indirect spatial confounding for the data in question using k-NN and distance based weights. First, the result showed that spatial interference and spatial confounding are two different phenomenon as the estimates are significantly different under both scenarios.  The result revealed that the estimates differs significantly. Of major concern is that, when direct spatial confounding and indirect spatial confounding are present but not accounted for, the effect estimates also varies significantly indicating that both are different from each other. The best model was found to be model which accounted for all the three phenomena due to its lowest AIC.  

\begin{figure}[H]
  \centering
\includegraphics[width=1\textwidth]{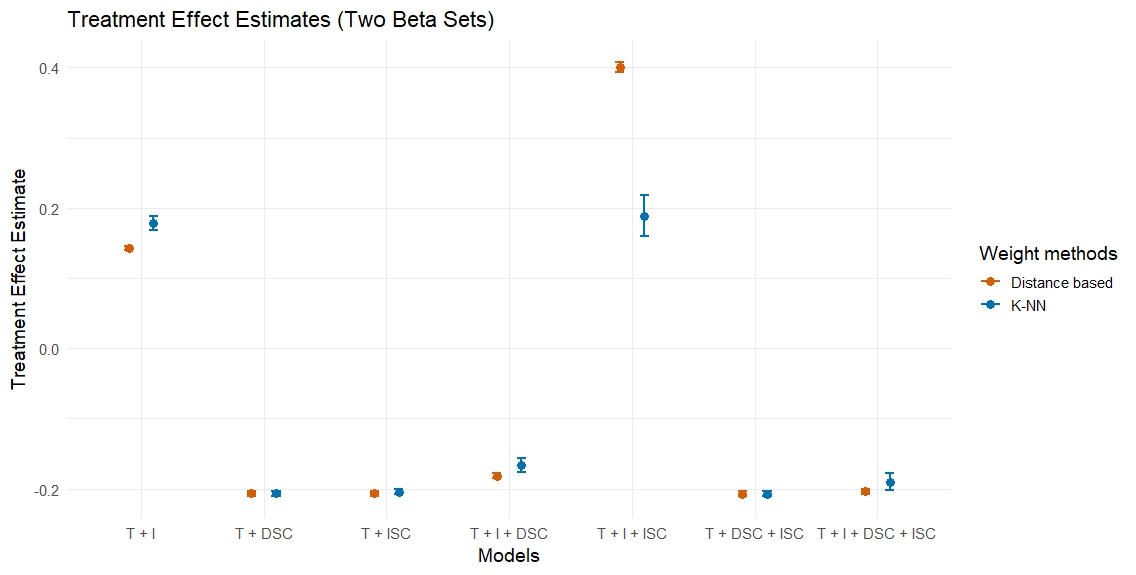}
   \caption{Plot of the treatment effect estimates for the models with 95\% confidence interval bounds}
  \label{app - data} 
\end{figure}

Figure \ref{app - data} shows the graphical representation of the treatment effect estimates. The blue dots represents the treatment effect estimates obtained using K-NN based weight for the construction of spatial interference and indirect spatial confounders while the red dots show the treatment effect estimates using distance based weights. The Figure shows clearly that spatial interference and spatial confounding are two different phenomenon as the estimates under both differs from each other. Also, the plot shows that direct spatial confounding and indirect spatial confounding are also the same as the estimates when only each of the phenomena is considered differs too. As a result, spatial confounding can be disentangled as contrary to assumption in existing literature that both are one. In addition, it can be deduced that when all are presents but any is ignored, the estimates varies significantly.

\section{Discussion}
The consequences of neglecting spatial interference when present can be severe or mild. This is dependent on the magnitude of the interference determined by the strength of spatial correlation among the treatment. Unlike existing literature which failed to clear the air on when to be bothered by spatial interference effect, this study revealed that the average bias is highest when interference is present but not accounted with highest RMSE and low coverage probability. This study also revealed that even when interference effect is accounted for, the distribution of the treatment contribute significantly to the magnitude of the bias. In fact, it can sometimes changes causal interpretation. Also, dealing with spatial or non spatial response constitutes to the bias. Bias tends to be higher in discrete treatment than in continuous treatment in non spatial setting while the opposite is observed in spatial setting. In short, the magnitude of the interference effect, the distribution of the treatment and the  spatial or non spatial nature of the response need to be scrutinized when dealing with causal inference on observational spatial data. 

In addition, k-NN and distance based weights were adopted in the construction of the weight matrix which is dependent on the spatial correlation among treatment. Changes in the number of nearest neighbours and the distance threshold or cut off as it affect the treatment estimates were also investigate and scrutinized in our study. The simulation results showed that, if the weight is constructed based on distance, the bias induced by the interference effect in both spatial and non spatial setting decreases as the threshold distance increases. In contrast, using k-NN weights, the bias increases as the number of nearest neighbours increases except for in spatial setting which behaves irregularly as shown in Appendices \ref{app: B1} and \ref{app: B2}. 

Furthermore, previous studies on spatial confounding bias assumed both the treatment and confounder to be normally distributed (e.g. \citet{khan2023}) which is not always the case. We studied the bias due to changes in the distributions of both the treatment and confounder. Our study revealed that the distributional properties of both the confounder and treatment is also a key factor when dealing with spatial confounding bias in causal inference.  The key finding here is that the bias changes significantly when the distribution of the treatment and confounders changes. Explicitly, our simulation result showed that the bias is greater when both the treatment and confounder are normally distributed. Strengthening our arguments and motivation for deriving the bias when both the confounder and treatment distributions changes which was not previously studied in existing causal inference literature.

Also, evaluating the bias in the presence of spatial interference, direct spatial confounding and indirect spatial confounding for normally distributed treatment and confounder using k-NN and distance based weights revealed that the bias is higher in models where interference effect is neglected with higher RMSE irrespective of the weighting methods. This is corroborating the importance of not neglecting interference effects in spatial causal inference when present. However, using k-NN weight based, the result showed that when interference effect and direct spatial confounding are only accounted for, the treatment estimate is more biased than when the three (interference effect, direct spatial confounding indirect spatial confounding) are accounted for. But in contrary using distance based weights, the bias are almost similar in both cases and falls within the same range. Hence, we conclude that indirect spatial confounding can not always be neglected but dependent on how the confounders interferes with responses and treatments in another location.  The deduction here is that, direct and and indirect confounding are not the same which is in contrary to previous studies which accounted for spatial confounding but assumed both direct and indirect spatial confounding to be one as studied by \citet{dupont2023} and \citet{khan2022}. Though, both can be interpreted as one when the indirect confounding is constructed in such away that makes the interference effects less  significant but this is not always the case. 

Finally, real data application on air surface temperature, land surface temperature, and precipitation revealed that spatial interference and spatial confounding are not the same and so also spatial confounding can be disentangle into direct and indirect spatial confounding. The model which accounted for the three phenomenon has the best fit to the data. In summary, when all are present and any is ignore, the treatment effect will be biased and inference made on such estimate may not be reliable.

\subsection{Bias-adjustment due to spatial interference}
Here, we present simple method to correct the bias induced by spatial interference based on the quantified bias. Assuming a known value of $\beta_{\tilde{a}}$, the bias can be corrected by subtracting the bias derived in equations \ref{SI non spatial} and \ref{SI  spatial} for the spatial and non spatial settings respectively from the estimated biased treatment effect $\hat{\beta_a}$. 
Hence, the corrected treatment effect for spatial and non spatial settings are given respectively as:
\begin{align*}
   \text{Corrected} \;\; \beta_a &= \hat{\beta_a} -  Bias\left(\hat{\beta_a^{*}}\right) \\
   &= \hat{\beta_a}  - \beta_{\tilde{a}}\frac{A^T \Psi A}{A^{T} A}  
\end{align*}
and 
\begin{align*}
   \text{Corrected} \;\; \beta_a &= \hat{\beta_a} -  Bias\left(\hat{\beta_a}\right) \\
   &= \hat{\beta_a}  - \beta_{\tilde{a}}\frac{A^T \Omega^{-1} \Psi A}{A^{T} \Omega^{-1}A}   
\end{align*}

\textbf{Algorithm for adjusting the bias}

The following steps can be taken in order to adjust for the bias:

(i)  Pick a $\Psi$ class and compute the bias

(ii) estimate  $\hat{\beta_a}$

(iii) compute adjusted $\hat{\beta_a}$

\subsection{Limitations and future study}

We have provided insight on the impacts of spatial interference bias when neglected in causal inference but we have only considered just two different weights. However, other forms of weights can also be explored. In addition, we distinguished between direct and indirect confounding and quantify the bias when spatial interference, direct and indirect spatial confounding are all present but we also did not consider bias due to spatial heterogeneity in causal inference. Also, we only gave a simple bias correction approach due to spatial interference bias based on quantified bias but we did not provide method for precisely estimating the treatment effect in the presence of both spatial interference and spatial confounder. In future, we intend to address heterogeneity in causal inference and developed a framework for simultaneously addressing spatial interference, direct and indirect spatial confounder for both direct and mediated effect estimation

\section{Conclusion}
In a novel way, we gave clarity to misconceptions and issues on spatial confounding from DAG perspective and disentangled spatial interference, direct and indirect spatial confounding. We also derived and evaluated the bias induced by spatial interference and indirect spatial confounding. The results showed that the choice of weights, the distributional properties and the magnitude of the interference effect are key determinants of the magnitude of bias due to spatial interference and the spatial covariance structure plays a key role in spatial setting.  
Also, the simulation results showed that, if the weight is constructed based on distance,  the bias induced by the interference effect in both spatial and non spatial setting decreases as the threshold distance increases. In contrast, using k-NN weights, the bias increases as the number of nearest neighbours increases except for in spatial setting which behaves irregularly.
On the other hand, we noted that the bias induced by spatial confounding (direct) previously obtained in literature 
was based on normality assumption for both the treatment and missing confounder. 
We pointed out that if the distribution of the confounder and treatment change, the bias will also change. We illustrated that using Poisson distributed confounder and treatment. 
Our simulation results and real life data application showed that direct and indirect spatial confounding are not the same. If spatial interference, direct and indirect spatial confounding are all present, and any is ignored the treatment effect of interest will be biased and inference made on such estimates can be misleading. 

\newpage
\appendix
\begin{center}
    \textbf{APPENDIX}
\end{center}

\section{Direct spatial confounding bias results}
\label{app:A}
Bias due to spatial confounding i.e.  $\beta_u \left(A^TA\right)^{-1}A^T E\left(U|A \right) $ in  non spatial setting was shown by \citet{paciorek2010} to be: $$\beta_u \rho \frac{\sigma_u}{\sigma_a}\left[ \left(A^{*T}A^*\right)^{-1} A^{*T}K \left(A - \mu_a1 \right)  \right]_2$$ 

Similarly for spatial setting, \citet{khan2023} derived the bias due to spatial confounding as:  $$ \beta_u \rho \frac{\sigma_u}{\sigma_a}\left[ \left(A^{*T} \Omega^{-1}A^*\right)^{-1} A^{*T} \Omega^{-1} K \left(A - \mu_a1 \right)  \right]_2$$  where $K = p_c \left(  p_c I + (1 - p_c)  \Omega(\omega_u) \Omega(\omega_c)^{-1} \right)^{-1}$, $p_c = \frac{\sigma^2}{ \sigma_c^2 + \sigma_u^2}$ and $A^* = [1 A]$. These derivations were adopted in section 4.4. See \citet{khan2023} for full derivation. 

\section{Supplementary simulation results}

\subsection{Simulation result when k = 1, 2, 3 and 5}
\label{app: B1}
\begin{table}[H]
\centering
\caption{Simulation result for discrete treatment in non spatial setting }
\label{Discrete Non spatial}
\begin{tiny}
\begin{tabular}{ll|ccc|ccc}
\hline \hline
Weights &  &   \multicolumn{3}{|c|}{k-NN (k=1)}  &  \multicolumn{3}{|c} {k-NN (k=2)}\\
\hline \hline
\textbf{Cases} & \textbf{Scenario} & \textbf{Bias} & \textbf{RMSE} & \textbf{Coverage} & \textbf{Bias} & \textbf{RMSE} & \textbf{Coverage}\\
\hline \hline
\multirow{2}{*}{$\beta_a = 8, \beta_{\hat{a}} = 2$} 
  & Scenario 1 & -0.0437 & 0.3178 &    0.832 & 0.3054 &  0.3611  &  0.705\\
  & Scenario 2 & -0.0056 & 0.1620 &   0.944 & -0.0008 & 0.1651 &   0.956 \\
\multirow{2}{*}{$(\beta_a > \beta_{\hat{a}})$} 
  & Scenario 3 & -0.0057 & 0.1600   & 0.948 & 0.0068 & 0.1586 &   0.950  \\
  & Scenario 4 & -0.0056 & 0.1620  &  0.944 &-0.0008 & 0.1651 &   0.956 \\
\hline \hline 
\multirow{2}{*}{$\beta_a = 3, \beta_{\hat{a}} = 9$} 
     & Scenario 1 & -0.1768 & 1.2610 &    0.744  & 1.3717 & 1.5190  &  0.497  \\
  & Scenario 2 &  -0.0056 & 0.1620  &   0.944 &0.0000 & 0.1651 &   0.956 \\ 
  \multirow{2}{*}{$\beta_a < \beta_{\hat{a}}$} 
   & Scenario 3 & -0.0057 & 0.1601   & 0.948 & 0.0007 &0.1586 &   0.950\\
  & Scenario 4 & -0.0056 & 0.1620  &  0.944 &0.0000 & 0.1651 &   0.956 \\
\hline \hline
 &  &   \multicolumn{3}{|c|}{k-NN (k=3)}  &  \multicolumn{3}{|c} {k-NN (k=5)}\\
\hline \hline
\multirow{2}{*}{$\beta_a = 8, \beta_{\hat{a}} = 2$} 
  & Scenario 1 & 0.4377 & 0.4607 &    0.544 & 0.5560 &  0.5641  &  0.366\\
  & Scenario 2 & 0.0015 & 0.1700 &   0.942 & -0.0112 & 0.1833 &   0.944 \\
\multirow{2}{*}{$(\beta_a > \beta_{\hat{a}})$} 
  & Scenario 3 & 0.0038 & 0.1574   & 0.953 & -0.0053 & 0.1611 &   0.953  \\
  & Scenario 4 & 0.0015 & 0.1700 &   0.942 &-0.0112 & 0.1833 &   0.944 \\
\hline \hline
\multirow{2}{*}{$\beta_a = 3, \beta_{\hat{a}} = 9$} 
     & Scenario 1 & 1.9564 & 1.9935 &    0.257  & 2.5204 & 2.5257  &  0.089  \\
  & Scenario 2 &  0.0015 & 0.1700 &   0.942 & -0.0112 & 0.1833 &   0.944 \\ 
  \multirow{2}{*}{$\beta_a < \beta_{\hat{a}}$} 
   & Scenario 3 & 0.0038 & 0.1574   & 0.953 & 0.0053 &0.1611 &   0.953\\
  & Scenario 4 & 0.0015 & 0.1700  &  0.942 &-0.0112 & 0.1833 &   0.944 \\
\hline \hline
\end{tabular}
\end{tiny}
\end{table}

\begin{table}[H]
\centering
\caption{Simulation results for continuous treatment in non spatial setting }
\label{continous Non spatial}
\begin{tiny}
\begin{tabular}{ll|ccc|ccc}
\hline \hline
Weights &  &   \multicolumn{3}{|c|}{k-NN (k=1)}  &  \multicolumn{3}{|c} {k-NN (k=2)}\\
\hline \hline
\textbf{Cases} & \textbf{Scenario} & \textbf{Bias} & \textbf{RMSE} & \textbf{Coverage} & \textbf{Bias} & \textbf{RMSE} & \textbf{Coverage}\\
\hline \hline
\multirow{2}{*}{$\beta_a = 8, \beta_{\hat{a}} = 2$} 
  & Scenario 1 & -0.0498 & 0.4753 &    0.926 & 0.4447 &  0.7074  &  0.905\\
  & Scenario 2 & -0.0173 & 0.4046 &   0.949 & -0.0121  & 0.5885 &   0.952 \\
\multirow{2}{*}{$(\beta_a > \beta_{\hat{a}})$} 
  & Scenario 3 & -0.0143 & 0.4027   & 0.947 & -0.0062  & 0.5566 &   0.951  \\
  & Scenario 4 & -0.0173 & 0.4046 &   0.949 & -0.0121 & 0.5885  &  0.952 \\
\hline \hline 
\multirow{2}{*}{$\beta_a = 3, \beta_{\hat{a}} = 9$} 
     & Scenario 1 & -0.1742 & 1.2414 &    0.796 & 2.0336 & 2.1311  &  0.478  \\
  & Scenario 2 &  -0.0173 & 0.4046 &   0.949 &-0.0121  & 0.5885 &   0.952 \\ 
  \multirow{2}{*}{$\beta_a < \beta_{\hat{a}}$} 
   & Scenario 3 & -0.0143 & 0.4027   & 0.947 & -0.0062  & 0.5566 &   0.951\\
  & Scenario 4 & -0.0173 & 0.4046 &   0.949 &-0.0121 & 0.5885  &  0.952 \\
\hline \hline
 &  &   \multicolumn{3}{|c|}{k-NN (k=3)}  &  \multicolumn{3}{|c} {k-NN (k=5)}\\
\hline \hline
\multirow{2}{*}{$\beta_a = 8, \beta_{\hat{a}} = 2$} 
  & Scenario 1 & 0.6472 & 0.9143 &    0.881 & 0.8030 &  1.1619  &  0.898\\
  & Scenario 2 & -0.0085 & 0.7701 &   0.951 & -0.0639 & 0.1693 &   0.946 \\
\multirow{2}{*}{$(\beta_a > \beta_{\hat{a}})$} 
  & Scenario 3 & 0.0000 & 0.6826   & 0.945 &-0.0131 &0.9171 &   0.943  \\
  & Scenario 4 & -0.0085 & 0.7701 &   0.951 &-0.0639 & 1.1693  &  0.946 \\
\hline \hline
\multirow{2}{*}{$\beta_a = 3, \beta_{\hat{a}} = 9$} 
     & Scenario 1 & 2.9122 & 2.9595 &    0.302  & 3.6595 & 3.6901 &  0.267  \\
  & Scenario 2 &  -0.0085 & 0.7701 &   0.951 &-0.0639 & 0.1693 &   0.946 \\ 
  \multirow{2}{*}{$\beta_a < \beta_{\hat{a}}$} 
   & Scenario 3 & 0.0000 & 0.6826   & 0.945 & -0.0131 &0.9171 &   0.943\\
  & Scenario 4 & -0.0085 & 0.7701 &   0.951&-0.0639 & 1.1693  &  0.946 \\
\hline \hline
\end{tabular}
\end{tiny}
\end{table}

\begin{table}[H]
\centering
\caption{Simulation result for discrete treatment in spatial setting }
\label{Discrete spatial}
\begin{tiny}
\begin{tabular}{ll|ccc|ccc}
\hline \hline
Weights &  &   \multicolumn{3}{|c|}{k-NN (k=1)}  &  \multicolumn{3}{|c} {k-NN (k=2)}\\
\hline \hline
\textbf{Cases} & \textbf{Scenario} & \textbf{Bias} & \textbf{RMSE} & \textbf{Coverage} & \textbf{Bias} & \textbf{RMSE} & \textbf{Coverage}\\
\hline \hline
\multirow{2}{*}{$\beta_a = 8, \beta_{\hat{a}} = 2$} 
  & Scenario 1 & -0.4240 & 0.518 &    0.506 & -0.3680 &  0.404  &  0.502\\
  & Scenario 2 & 0.0077 & 0.134 &   0.947 &0.0065 & 0.135 &   0.946 \\
\multirow{2}{*}{$(\beta_a > \beta_{\hat{a}})$} 
  & Scenario 3 & 0.0044 & 0.127   & 0.943 &0.0044 &0.127 &   0.943  \\
  & Scenario 4 & 0.0077 & 0.134 &   0.947 &0.0065 & 0.135 &   0.946 \\
\hline \hline 
\multirow{2}{*}{$\beta_a = 3, \beta_{\hat{a}} = 9$} 
     & Scenario 1 & -1.0200 & 2.680 &    0.348  & -1.8200 & 1.900  &  0.236  \\
  & Scenario 2 & 0.0077 & 0.134 &   0.947 & 0.0065 & 0.135 &   0.946 \\ 
  \multirow{2}{*}{$\beta_a < \beta_{\hat{a}}$} 
   & Scenario 3 & 0.0044 & 0.127   & 0.943 & 0.0044 &0.127 &   0.943\\
  & Scenario 4 & 0.0077 & 0.134 &   0.947 & 0.0065 & 0.135  &  0.946 \\
\hline \hline
 &  &   \multicolumn{3}{|c|}{k-NN (k=3)}  &  \multicolumn{3}{|c} {k-NN (k=5)}\\
\hline \hline
\multirow{2}{*}{$\beta_a = 8, \beta_{\hat{a}} = 2$} 
  & Scenario 1 & -0.3030 & 0.333 &    0.551 & -0.2200 &  0.251  &  0.679\\
  & Scenario 2 & 0.0067 & 0.136 &   0.943 &0.0083 & 0.134 &   0.943 \\
\multirow{2}{*}{$(\beta_a > \beta_{\hat{a}})$} 
  & Scenario 3 & 0.0044 & 0.127   & 0.943 &0.0044 &0.127 &   0.943  \\
  & Scenario 4 & 0.0067 & 0.136  &  0.943 &0.0083 & 0.134  &  0.943 \\
\hline \hline
\multirow{2}{*}{$\beta_a = 3, \beta_{\hat{a}} = 9$} 
     & Scenario 1 & -1.580 & 1.610 &    0.167  & -1.200 & 1.210  &  0.133  \\
  & Scenario 2 & 0.0067 & 0.136 &   0.943 &0.0083 & 0.134 &   0.943 \\ 
  \multirow{2}{*}{$\beta_a < \beta_{\hat{a}}$} 
   & Scenario 3 & 0.0044 & 0.127   & 0.943 & 0.0044 &0.127 &   0.943\\
  & Scenario 4 & 0.0067 & 0.136 &   0.943 &0.0083 & 0.134 &   0.943 \\
\hline \hline
\end{tabular}
\end{tiny}
\end{table}

\begin{table}[H]
\centering
\caption{Simulation result for continuous treatment in spatial setting }
\label{continuous spatial}
\begin{tiny}
\begin{tabular}{ll|ccc|ccc}
\hline \hline
Weights &  &   \multicolumn{3}{|c|}{k-NN (k=1)}  &  \multicolumn{3}{|c} {k-NN (k=2)}\\
\hline \hline
\textbf{Cases} & \textbf{Scenario} & \textbf{Bias} & \textbf{RMSE} & \textbf{Coverage} & \textbf{Bias} & \textbf{RMSE} & \textbf{Coverage}\\
\hline \hline
\multirow{2}{*}{$\beta_a = 8, \beta_{\hat{a}} = 2$} 
  & Scenario 1 & 0.2490 & 0.4040 &    0.655  & -0.0300 &  0.2060  &  0.818\\
  & Scenario 2 & 0.0054 & 0.0826 &   0.937 &0.0049 & 0.0825 &   0.938 \\
\multirow{2}{*}{$(\beta_a > \beta_{\hat{a}})$} 
  & Scenario 3 & 0.0054 & 0.0827 &   0.936 &0.0054 &0.0827 &   0.936 \\
  & Scenario 4 & 0.0054 & 0.0826 &   0.937 &0.0049 & 0.0825 &   0.938 \\
\hline \hline 
\multirow{2}{*}{$\beta_a = 3, \beta_{\hat{a}} = 9$} 
     & Scenario 1 & 2.2600 & 2.9100 &    0.522  & -0.2860 & 0.8760  &  0.774  \\
  & Scenario 2 &0.0054 & 0.0826 &   0.937  &0.0049 & 0.0825 &   0.938 \\ 
  \multirow{2}{*}{$\beta_a < \beta_{\hat{a}}$} 
   & Scenario 3 & 0.0054 & 0.0827 &   0.936 &0.0054 &0.0827 &   0.936 \\
  & Scenario 4 & 0.0054 & 0.0826 &   0.937 &0.0049 & 0.0825 &   0.938 \\
\hline \hline
 &  &   \multicolumn{3}{|c|}{k-NN (k=3)}  &  \multicolumn{3}{|c} {k-NN (k=5)}\\
\hline \hline
\multirow{2}{*}{$\beta_a = 8, \beta_{\hat{a}} = 2$} 
  & Scenario 1 & -0.0765 & 0.1710 &    0.833 & -0.0891 &  0.148  &  0.826\\
  & Scenario 2 & 0.0046 & 0.0828 &   0.940 &0.0051 & 0.0829 &   0.943 \\
\multirow{2}{*}{$(\beta_a > \beta_{\hat{a}})$} 
  & Scenario 3 & 0.0054 & 0.0827   & 0.936 &0.0054 &0.0827 &   0.936  \\
  & Scenario 4 & 0.0046 & 0.0828 &   0.940 &0.0051 & 0.0829 &   0.943 \\
\hline \hline
\multirow{2}{*}{$\beta_a = 3, \beta_{\hat{a}} = 9$} 
     & Scenario 1 & -0.5320 & 0.7110 &    0.755  & -0.5630 & 0.632  &  0.647  \\
  & Scenario 2 &  0.0046 & 0.0828 &   0.940  &0.0051 & 0.0829 &   0.943 \\ 
  \multirow{2}{*}{$\beta_a < \beta_{\hat{a}}$} 
   & Scenario 3 & 0.0054 & 0.0827   & 0.936 & 0.0054 &0.0827 &   0.936\\
  & Scenario 4 & 0.0046 & 0.0828 &   0.940  &0.0051 & 0.0829 &   0.943 \\
\hline \hline
\end{tabular}
\end{tiny}
\end{table}

\subsection{Simulation results when distance  threshold = 90\%, 80\%, 75\% and 50\%}
\label{app: B2}

\begin{table}[H]
\centering
\caption{Simulation result for discrete treatment in non spatial setting }
\label{Discrete Non spatial - distance}
\begin{tiny}
\begin{tabular}{ll|ccc|ccc}
\hline \hline
Weights &  &   \multicolumn{3}{|c|}{90\% threshold}  &  \multicolumn{3}{|c} {80\% threshold}\\
\hline \hline
\textbf{Cases} & \textbf{Scenario} & \textbf{Bias} & \textbf{RMSE} & \textbf{Coverage} & \textbf{Bias} & \textbf{RMSE} & \textbf{Coverage}\\
\hline \hline
\multirow{2}{*}{$\beta_a = 8, \beta_{\hat{a}} = 2$} 
  & Scenario 1 & 0.0036 & 0.1703 &    0.936 & 0.0855 &  0.1961  &  0.892\\
  & Scenario 2 & -0.0007 & 0.1750 &   0.945 & -0.0064 & 0.1787 &   0.944 \\
\multirow{2}{*}{$(\beta_a > \beta_{\hat{a}})$} 
  & Scenario 3 & 0.0000 & 0.1629   & 0.947 & -0.0039 &0.1586 &   0.953  \\
  & Scenario 4 &  -0.0007 & 0.1750 &   0.945 & -0.0064 & 0.1787 &   0.944\\
\hline \hline 
\multirow{2}{*}{$\beta_a = 3, \beta_{\hat{a}} = 9$} 
     & Scenario 1 & 0.0163 & 0.2805 &    0.772  & 0.3980 & 0.5572  &  0.492  \\
  & Scenario 2 &  -0.0007 & 0.1750 &   0.945 & -0.0064 & 0.1787 &   0.944 \\ 
  \multirow{2}{*}{$\beta_a < \beta_{\hat{a}}$} 
   & Scenario 3 & 0.0000 & 0.1639   & 0.947 & -0.0039 & 0.1586 &   0.953\\
  & Scenario 4 & -0.0007 & 0.1750 &   0.945 &-0.0064 & 0.1787  &  0.944 \\
\hline \hline
 &  &   \multicolumn{3}{|c|}{75\% threshold}  &  \multicolumn{3}{|c} {50\% threshold}\\
\hline \hline
\multirow{2}{*}{$\beta_a = 8, \beta_{\hat{a}} = 2$} 
  & Scenario 1 &  0.1527 & 0.2406 &    0.829 & 0.4451 &  0.4789  &  0.450\\
  & Scenario 2 & 0.0046 & 0.1903 &   0.936 & 0.0131 & 0.2034 &   0.947 \\
\multirow{2}{*}{$(\beta_a > \beta_{\hat{a}})$} 
  & Scenario 3 & 0.0074 & 0.1615   & 0.943 &0.0090 &0.1591 &   0.949  \\
  & Scenario 4 & 0.0046 & 0.1903 &   0.936 &0.0131 & 0.2034 &   0.947 \\
\hline \hline
\multirow{2}{*}{$\beta_a = 3, \beta_{\hat{a}} = 9$} 
     & Scenario 1 & 0.6615 & 0.7741 &    0.375  & 1.9719 & 2.0205  &  0.159  \\
  & Scenario 2 &  0.0046 & 0.1903 &   0.936 & 0.0131 & 0.2034 &   0.947 \\ 
  \multirow{2}{*}{$\beta_a < \beta_{\hat{a}}$} 
   & Scenario 3 & 0.0074 & 0.1615   & 0.943 & 0.0090 & 0.1591 &   0.949\\
  & Scenario 4 & 0.0046 & 0.1903  &  0.936 & 0.0131 & 0.2034  &  0.947 \\
\hline \hline
\end{tabular}
\end{tiny}
\end{table}

\begin{table}[H]
\centering
\caption{Simulation result for continuous treatment in non spatial setting }
\label{continuous Non spatial - distance}
\begin{tiny}
\begin{tabular}{ll|ccc|ccc}
\hline \hline
Weights &  &   \multicolumn{3}{|c|}{90\% threshold}  &  \multicolumn{3}{|c} {80\% threshold}\\
\hline \hline
\textbf{Cases} & \textbf{Scenario} & \textbf{Bias} & \textbf{RMSE} & \textbf{Coverage} & \textbf{Bias} & \textbf{RMSE} & \textbf{Coverage}\\
\hline \hline
\multirow{2}{*}{$\beta_a = 8, \beta_{\hat{a}} = 2$} 
  & Scenario 1 & 0.1003 & 11.7803 &    0.963 &  0.5100 &  8.4193  &  0.956\\
  & Scenario 2 & -0.3877 & 13.6483 &   0.953 & 0.0223 & 10.7109 &   0.950 \\
\multirow{2}{*}{$(\beta_a > \beta_{\hat{a}})$} 
  & Scenario 3 & 0.0510 & 11.7832   & 0.963 & 0.3292 & 8.4155 &   0.955  \\
  & Scenario 4 & -0.3877 & 13.6483 &   0.953 & 0.0223 & 10.7108  &  0.950 \\
\hline \hline 
\multirow{2}{*}{$\beta_a = 3, \beta_{\hat{a}} = 9$} 
     & Scenario 1 & 0.2731 & 11.7739 &    0.963  & 1.1425 & 8.4925  &  0.954  \\
  & Scenario 2 &  -0.3877 & 13.6483 &   0.953 &0.0223 & 10.7109 &   0.950 \\ 
  \multirow{2}{*}{$\beta_a < \beta_{\hat{a}}$} 
   & Scenario 3 & 0.0510 & 11.7832   & 0.963 &0.3292 & 8.4155 &   0.955\\
  & Scenario 4 &-0.3877 & 13.6483 &   0.953 &0.0223 & 10.7108  &  0.950 \\
\hline \hline
 &  &   \multicolumn{3}{|c|}{75\% threshold}  &  \multicolumn{3}{|c} {50\% threshold}\\
\hline \hline
\multirow{2}{*}{$\beta_a = 8, \beta_{\hat{a}} = 2$} 
  & Scenario 1 & 0.6089 & 7.6306 &  0.949 & 0.7649 &  4.3006  &  0.951\\
  & Scenario 2 & 0.0609 & 10.0416 &   0.938 & 0.2478 & 6.1634 &   0.953 \\
\multirow{2}{*}{$(\beta_a > \beta_{\hat{a}})$} 
  & Scenario 3 & 0.3603 & 7.6179   & 0.949 &0.1818 & 4.2446 &   0.954  \\
  & Scenario 4 & 0.0609 & 10.0416 &   0.938 & 0.2478 & 6.1634  &  0.953 \\
\hline \hline
\multirow{2}{*}{$\beta_a = 3, \beta_{\hat{a}} = 9$} 
     & Scenario 1 & 1.4787 & 7.7689 &    0.951  & 2.808 & 5.1064  &  0.881  \\
  & Scenario 2 &  0.0609 & 10.0416 &   0.938 & 0.2478 & 6.1634 &   0.953 \\ 
  \multirow{2}{*}{$\beta_a < \beta_{\hat{a}}$} 
   & Scenario 3 & 0.3603 & 7.6179   & 0.949 & 0.1818 & 4.2446 &   0.954\\
  & Scenario 4 & 0.0609 & 10.0416 &   0.938 &0.2478 & 6.1634  &  0.953 \\
\hline \hline
\end{tabular}
\end{tiny}
\end{table}

\begin{table}[H]
\centering
\caption{Simulation result for discrete treatment in spatial setting }
\label{Discrete spatial - distance}
\begin{tiny}
\begin{tabular}{ll|ccc|ccc}
\hline \hline
Weights &  &   \multicolumn{3}{|c|}{90\% threshold}  &  \multicolumn{3}{|c} {80\% threshold}\\
\hline \hline
\textbf{Cases} & \textbf{Scenario} & \textbf{Bias} & \textbf{RMSE} & \textbf{Coverage} & \textbf{Bias} & \textbf{RMSE} & \textbf{Coverage}\\
\hline \hline
\multirow{2}{*}{$\beta_a = 8, \beta_{\hat{a}} = 2$} 
  & Scenario 1 & -0.0173 & 0.129 &    0.945 & -0.0186 &  0.129  &  0.944\\
  & Scenario 2 & 0.0056 & 0.138 &   0.945 &0.0063 & 0.133 &   0.948 \\
\multirow{2}{*}{$(\beta_a > \beta_{\hat{a}})$} 
  & Scenario 3 & 0.0044 & 0.127   & 0.943 &0.0044 &0.127 &   0.943  \\
  & Scenario 4 & 0.0058 & 0.138  &  0.945  & 0.0063 & 0.133  &  0.948 \\
\hline \hline 
\multirow{2}{*}{$\beta_a = 3, \beta_{\hat{a}} = 9$} 
     & Scenario 1 & -0.0949 & 0.153 &    0.896  & -0.1060 & 0.160  &  0.882  \\
  & Scenario 2 &  0.0058 & 0.138 &   0.945 &0.0063 & 0.133 &   0.948 \\ 
  \multirow{2}{*}{$\beta_a < \beta_{\hat{a}}$} 
   & Scenario 3 & 0.0044 & 0.127   & 0.943 & 0.0044 &0.127 &   0.943\\
  & Scenario 4 &  0.0058 & 0.138 &   0.945 & 0.0063 & 0.133  &  0.948 \\
\hline \hline
 &  &   \multicolumn{3}{|c|}{75\% threshold}  &  \multicolumn{3}{|c} {50\% threshold}\\
\hline \hline
\multirow{2}{*}{$\beta_a = 8, \beta_{\hat{a}} = 2$} 
  & Scenario 1 & -0.0196 & 0.129 &    0.945 & -0.0286 &  0.131  &  0.946\\
  & Scenario 2 & 0.0058 & 0.134 &   0.945 &0.0023 & 0.132 &   0.944 \\
\multirow{2}{*}{$(\beta_a > \beta_{\hat{a}})$} 
  & Scenario 3 & 0.0044 & 0.127   & 0.943 &0.0044 &0.127 &   0.943  \\
  & Scenario 4 & 0.0058 & 0.134 &   0.945 &0.0023 & 0.132 &   0.944 \\
\hline \hline
\multirow{2}{*}{$\beta_a = 3, \beta_{\hat{a}} = 9$} 
     & Scenario 1 & -0.115 & 0.163 &    0.879  & -0.182 & 0.208  &  0.775  \\
  & Scenario 2 &  0.0058 & 0.134 &   0.945 & 0.0023 & 0.132 &   0.944 \\ 
  \multirow{2}{*}{$\beta_a < \beta_{\hat{a}}$} 
   & Scenario 3 & 0.0044 & 0.127   & 0.943 & 0.0044 &0.127 &   0.943\\
  & Scenario 4 & 0.0058 & 0.134  &  0.945 &-0.0023 & 0.132  &  0.944 \\
\hline \hline
\end{tabular}
\end{tiny}
\end{table}

\begin{table}[H]
\centering
\caption{Simulation result for continuous treatment in spatial setting }
\label{Continuous spatial - Distance}
\begin{tiny}
\begin{tabular}{ll|ccc|ccc}
\hline \hline
Weights &  &   \multicolumn{3}{|c|}{90\% threshold}  &  \multicolumn{3}{|c} {80\% threshold}\\
\hline \hline
\textbf{Cases} & \textbf{Scenario} & \textbf{Bias} & \textbf{RMSE} & \textbf{Coverage} & \textbf{Bias} & \textbf{RMSE} & \textbf{Coverage}\\
\hline \hline
\multirow{2}{*}{$\beta_a = 8, \beta_{\hat{a}} = 2$} 
  & Scenario 1 & -0.0143 & 0.0834 &    0.932 & -0.0152 &  0.0847  &  0.935\\
  & Scenario 2 & 0.0063 & 0.0852 &   0.933 &0.0057 & 0.0842 &   0.932 \\
\multirow{2}{*}{$(\beta_a > \beta_{\hat{a}})$} 
  & Scenario 3 & 0.0054 & 0.0827   & 0.936 &0.0054 &0.0827 &   0.936  \\
  & Scenario 4 & 0.0063 & 0.0852 &   0.933 & 0.0057 & 0.0842  &  0.932 \\
\hline \hline 
\multirow{2}{*}{$\beta_a = 3, \beta_{\hat{a}} = 9$} 
     & Scenario 1 & -0.0962 & 0.121 &    0.825  & -0.1170 & 0.135  &  0.800  \\
  & Scenario 2 &  0.0063 & 0.0852 &   0.933 &0.0057 & 0.0842 &   0.932\\ 
  \multirow{2}{*}{$\beta_a < \beta_{\hat{a}}$} 
   & Scenario 3 & 0.0054 & 0.0827   & 0.936 & 0.0054 &0.0827 &   0.936\\
  & Scenario 4 & 0.0063 & 0.0852 &   0.933 &0.0057 & 0.0842 &   0.932 \\
\hline \hline
 &  &   \multicolumn{3}{|c|}{75\% threshold}  &  \multicolumn{3}{|c} {50\% threshold}\\
\hline \hline
\multirow{2}{*}{$\beta_a = 8, \beta_{\hat{a}} = 2$} 
  & Scenario 1 & -0.0151 & 0.0852 &    0.938 & -0.0238 & 0.0884 &    0.922\\
  & Scenario 2 & 0.0055 & 0.0848 &   0.940 & 0.0033 & 0.0843 &   0.936 \\
\multirow{2}{*}{$(\beta_a > \beta_{\hat{a}})$} 
  & Scenario 3 & 0.0054 & 0.0827   & 0.936 &0.0054 & 0.0827   & 0.936  \\
  & Scenario 4 & 0.0055 & 0.0848 &   0.940 &0.0033 & 0.0843 &   0.936 \\
\hline \hline
\multirow{2}{*}{$\beta_a = 3, \beta_{\hat{a}} = 9$} 
     & Scenario 1 & -0.1230 & 0.1410 &    0.787  & -0.1970 & 0.207  &  0.642  \\
  & Scenario 2 &   0.0055 & 0.0848 &   0.940 &0.0033 & 0.0843 &   0.936  \\ 
  \multirow{2}{*}{$\beta_a < \beta_{\hat{a}}$} 
   & Scenario 3 & 0.0054 & 0.0827   & 0.936 & 0.0054 & 0.0827   & 0.936\\
  & Scenario 4 & 0.0055 & 0.0848 &   0.940 &0.0033 & 0.0843 &   0.936  \\
\hline \hline
\end{tabular}
\end{tiny}
\end{table}

\subsection{Plot of the simulated confounder and treatment for normal, Poisson and binomial distributions with their corresponding outcomes}
\label{app: B3}
\begin{figure}[H]
  \centering
\includegraphics[width=1.0\textwidth]{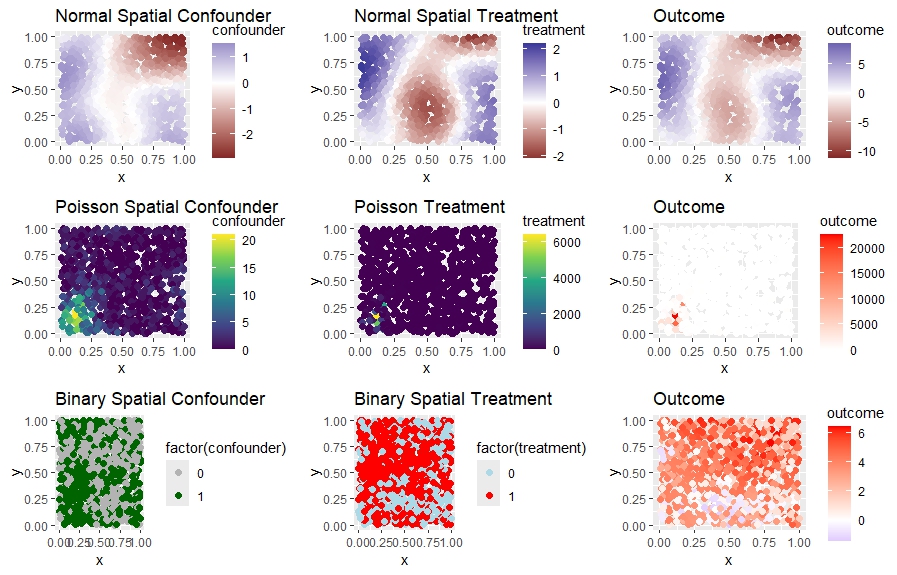}
   \caption{Normal, Poisson and  Binary spatial treatment and spatial confounders with corresponding outcomes}
  \label{Confounder and treaments} 
\end{figure}

\subsection{The bias box plots for binary treatment and confounder, normal treatment and confounder and Poisson treatment and confounder}
\label{app: B4}
\begin{figure}[H]
  \centering
\includegraphics[width=1.1\textwidth]{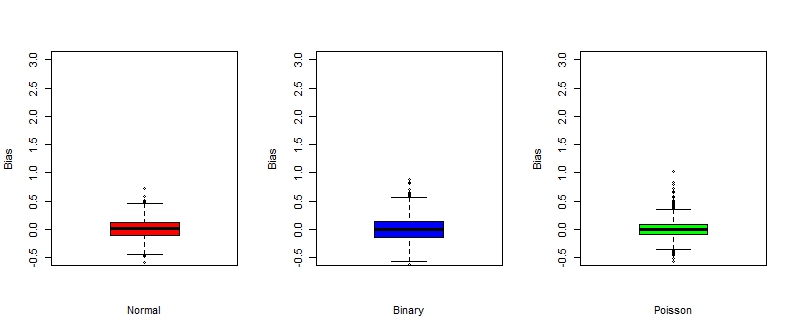}
   \caption{Bias for normal, Poisson, binary treatment and confounders}
  \label{Treatment and Confounders Bias with RMSE} 
\end{figure}

\bibliographystyle{apalike}
\bibliography{sources}

\end{document}